\UseRawInputEncoding
\documentclass[UTF8, 10pt,journal]{IEEEtran}
\IEEEoverridecommandlockouts
\usepackage{amsmath}
\usepackage{textcomp}
\usepackage[colorlinks, linkcolor=black,
anchorcolor=black,citecolor=black]{hyperref}
\usepackage{makecell}
\usepackage{balance}
\usepackage{bm}
\usepackage{amsmath,amssymb,amsfonts}
\usepackage{booktabs}
\usepackage{array}
\usepackage{caption}
\usepackage{graphicx}
\usepackage{algorithm}
\usepackage{algorithmic}
\usepackage{pifont}
\usepackage{harpoon}
\usepackage{multirow}
\usepackage{algorithm,float}
\usepackage{verbatim}
\usepackage{amsmath}
\usepackage{amsthm}
\newtheorem{theorem}{Theorem}

\usepackage{xcolor}
\usepackage{color}
\usepackage{subfigure}
\usepackage{threeparttable}
\usepackage{diagbox}
\usepackage{bbding}
\usepackage{pifont}
\usepackage{wasysym}
\usepackage{enumerate}
\usepackage{enumitem}
\usepackage{eqparbox}

\def\BibTeX{{\rm B\kern-.05em{\sc i\kern-.025em b}\kern-.08em
    T\kern-.1667em\lower.7ex\hbox{E}\kern-.125emX}}
\begin{document}

\title{SIMC 2.0: Improved Secure ML Inference \\ Against Malicious Clients }

\author{Guowen~Xu, Xingshuo~Han, Tianwei~Zhang, Shengmin~Xu, Jianting~Ning, Xinyi~Huang,  Hongwei~Li and Robert H.Deng~\IEEEmembership{Fellow,~IEEE}

\IEEEcompsocitemizethanks { \IEEEcompsocthanksitem Guowen~Xu, Xingshuo~Han and Tianwei~Zhang are with the School of Computer Science and Engineering, Nanyang Technological University. (e-mail: guowen.xu@ntu.edu.sg; xingshuo001@e.ntu.edu.sg; tianwei.zhang@ntu.edu.sg)
\IEEEcompsocthanksitem Shengmin~Xu  and Jianting~Ning are with the College of Computer and Cyber Security, Fujian
Normal University, Fuzhou, China (e-mail: smxu1989@gmail.com; jtning88@gmail.com)
\IEEEcompsocthanksitem Xinyi~Huang is with the Artificial Intelligence Thrust, Information Hub, Hong Kong University of Science and Technology (Guangzhou), Guangzhou, China, 511458 (e-mail: xinyi@ust.hk)
\IEEEcompsocthanksitem Hongwei~Li  is with the School of Computer Science and Engineering,  University of
Electronic Science and Technology of China, Chengdu 611731, China.(e-mail: hongweili@uestc.edu.cn)
\IEEEcompsocthanksitem  Robert H~Deng is with the School of Information Systems
Singapore Management University. (e-mail: robertdeng@smu.edu.sg) }}

 \IEEEcompsoctitleabstractindextext{
\begin{abstract}
\renewcommand{\raggedright}{\leftskip=0pt \rightskip=0pt plus 0cm}
 \raggedright
 In this paper, we study the problem of secure ML inference against a malicious client and a semi-trusted server such that the client only learns the inference output while the server learns nothing.  This problem is first formulated by Lehmkuhl \textit{et al.} with a solution (MUSE, Usenix Security'21), whose performance is then substantially improved by Chandran \textit{et al.}'s work (SIMC, USENIX Security'22). However, there still exists a nontrivial gap in these efforts towards practicality, giving the challenges of overhead reduction and secure inference acceleration in an all-round way.

  We propose SIMC 2.0, which complies with the underlying structure of SIMC, but significantly optimizes both the linear and non-linear layers of the model. Specifically, (1) we design a new coding method for homomorphic parallel computation between matrices and vectors. It is custom-built through the insight into the complementarity between cryptographic primitives in SIMC. As a result, it can minimize the number of rotation operations incurred  in the calculation process, which is very computationally expensive compared to other homomorphic operations (\textit{e.g.}, addition, multiplication). (2) We reduce the size of the garbled circuit (GC) (used to calculate nonlinear activation functions, \textit{e.g.}, ReLU) in SIMC by about two thirds. Then, we design an alternative lightweight protocol to perform tasks that are originally allocated to the expensive GCs. Compared with SIMC, our experiments show that SIMC 2.0 achieves a significant speedup by up to $17.4\times $ for linear layer computation, and at least $1.3\times$ reduction of both the computation and communication overheads in the implementation of non-linear layers under different data dimensions. Meanwhile, SIMC 2.0 demonstrates an encouraging runtime boost by $2.3\sim 4.3\times$ over SIMC on different state-of-the-art ML models.

\end{abstract}
}
\maketitle

\IEEEdisplaynotcompsoctitleabstractindextext

\IEEEpeerreviewmaketitle

The widespread application of machine learning (ML), especially the popularization of prediction services over pre-trained models, has increased the demand for \textit{secure inference}. In this process, a server $S_0$ holds an ML model $M$ whose weight $W$ is considered private and sensitive, while a client $S_1$ holds a private input $t$. The goal of secure inference  is to make $S_1$ only get the model's output while  $S_0$ knows nothing. Such a privacy-preserving paradigm has a variety of potential prospects, especially for privacy-critical applications, \textit{e.g.}, medical diagnosis, financial data analysis.
In theory, secure inference can be implemented by the secure two-party computing (2-PC) protocol in cryptographic primitives \cite{katz2018optimizing,wang2017faster}. It enables two parties to run an interactive protocol to securely calculate any function without revealing each party's private information. To instantiate it, a lot of impressive works \cite{rathee2020cryptflow2,patra2021aby2, mohassel2017secureml} have been proposed, and they are built on various cryptographic technologies such as homomorphic encryption (HE), secret sharing, and Yao's garbled circuits (GC). Due to the inherent complexity of cryptographic primitives, existing efforts exclusively focus on improving the efficiency and lie on a relatively weak threat assumption, \textit{i.e.}, the \textit{semi-honest adversary model} \cite{lavigne2018topology,hussain2021coinn,lou2021hemet}. In this model,  $S_0$ and $S_1$ faithfully follow the specifications of secure inference and can only capture private information through passive observations.

\subsection{Related Works}
\label{Related Works}
Recent works \cite{lehmkuhl2021muse,chandran2021simc,patrablaze} show that this \textit{semi-honest adversary model} can be insecure in real-world applications. Particularly, Lehmkuhl \textit{et al.} \cite{lehmkuhl2021muse} argue that it may be reasonable the ML model is held by a semi-trusted server, since the server is usually fixed and maintained by a reputable entity. However, it is impractical to assume that thousands of clients (which can be arbitrary entities) will faithfully comply with the protocol specification. To validate this hazard, they provide a systematic attack strategy, which enables a malicious client to violate the rules of the protocol and completely reconstruct the  model's parameters during the inference interactions. Moreover, the number of inference queries required for such an attack is far smaller than the most advanced black-box model extraction attack.

While the above problem  can be solved by resorting to the traditional 2-PC protocol against malicious adversaries \cite{damgaard2019new,escudero2020improved,hazay2019constant}, it is inefficient in practice. To bridge this gap,
Lehmkuhl \textit{et al.} \cite{lehmkuhl2021muse} pioneer the definition of \textit{client-malicious adversary model}, where the server is still considered semi-honest but the client can perform arbitrary malicious behaviors. To secure the inference process under such adversary model, they propose MUSE, a novel 2-PC protocol with considerably lower overhead compared to previous works. However, the computation and communication costs of MUSE are still not satisfactory, which are at least $15\times$ larger than the similar work DELPHI \cite{mishra2020delphi} under the semi-trusted  threat model.

To alleviate the efficiency issue, Chandran \textit{et al.} design SIMC \cite{chandran2021simc}, the state-of-the-art 2-PC protocol for secure inference under the client-malicious adversary model. Consistent with the underlying tone of MUSE, SIMC uses HE to execute the linear layers (including matrix-vector and convolution multiplication) of the ML model, and uses GC to implement the non-linear layers (mainly the ReLU function). Since almost $99\%$ of the communication overhead in MUSE comes from non-linear layers, the core of SIMC is a novel protocol based on a customized GC to improve the performance of executing non-linear activation functions. As a result, compared with MUSE, SIMC gains $28\sim33\times$  reduction in the communication overhead for the implementation of popular activation functions (\textit{e.g.}, ReLU, ReLU6), and at least $4\times$ overall performance improvement.

However, we argue that SIMC is still far from practicality. This stems from two reasons. First, SIMC keeps
the design of linear layer calculations in MUSE, which uses computationally heavy HE to perform dense matrix-vector and convolution multiplication. Actually, linear operations dominate the computation of modern neural networks: nearly $95\%$ of ML model execution is for intensive convolutional layers and fully connected (FC) layers \cite{zhang2021gala}. This raises the requirement for efficient execution of matrix and convolution multiplication under ciphertext. Although the homomorphism of HE makes it suitable to achieve privacy-preserving linear operations, it is still computationally expensive for large-scale operations, especially when there is no appropriate parallel computing optimization (\textit{i.e.}, performing
homomorphic linear computation in a Single Instruction Multiple Data (SIMD) manner) \cite{sav2020poseidon}.

There is still much optimization space for the non-linear layers (such as ReLU) of SIMC, from the perspectives of both computation and communication. To be precise, given secret shares of a value $t$, SIMC designs a GC-based secure 2-PC protocol, which enables $S_0$ and $S_1$ to calculate shares $s=\alpha t$ and  authenticated shares of $v=f(t)$, \textit{i.e.}, shares of $f(t)$ and $\alpha f(t)$, where $f$ denotes the activation function. While SIMC has made great efforts to simplify the design of the GC and prevent malicious behaviors from the client, it  requires the communication overhead of at least $2c\lambda+4\kappa\lambda+6\kappa^2$, where $\lambda$ denotes the security parameter, $\kappa$ is a field space, and $c$  is the number of AND gates required to reconstruct shares of $t$ and compute authenticated $f(t)$. Besides, the main body of the activation function is still sealed in a relatively complex GC. As a consequence, SIMC inevitably results in substantial computation and communication costs in non-linear layers,  since a modern model usually contains thousands of non-linear activation functions for calculation. We will perform experiments to demonstrate such overheads in Section~\ref{sec:PERFORMANCE EVALUATION}.

\subsection{Technical Challenges}
\label{Technical Challenges}
This paper is dedicated to design a new 2-PC protocol to break through the performance bottleneck in SIMC, thereby promoting the practicality of secure inference against malicious clients. We follow the underlying structure of SIMC, \textit{i.e.}, using HE to execute the linear layers and GC to implement the non-linear layers, as such a hybrid method  has shown advanced performance compared to other strategies \cite{lehmkuhl2021muse,mishra2020delphi}. Therefore, our work naturally lie in solving two problems:  (1) how to design an optimized mechanism to accelerate linear operations in HE, and (2) how to find a more simplified GC for the execution of non-linear activation functions.

There have been many impressive works \cite{huang2022cheetah,jiang2018secure,sav2020poseidon,zhang2021gala} exploring methods towards the above goals. To speed up HE's computation performance, existing efforts mainly focus on designing new coding methods  to achieve parallelized component-wise
homomorphic computation, \textit{i.e.}, performing homomorphic linear computation in an SIMD manner. For example, Jiang \textit{et al.} \cite{jiang2018secure} present a novel matrix encoding method for basic matrix operations, \textit{e.g.},  multiplication and transposition.  Compared with previous approaches, this method reduces the computation complexity of multiplying two $d\times d$ matrices from $O(d^3)$ to $O(d)$. However, it is exclusively applicable to matrix operations between square or rectangular matrices, but unfriendly to arbitrary matrix-vector multiplication and convolution operations in ML. In addition, it considers parallel homomorphic calculations in a full ciphertext environment, contrary to our scenario where only the input of the client is ciphertext while the model parameters are clear.  Sav \textit{et al.}\cite{sav2020poseidon} propose the Alternating Packing (AP) approach, which packs all plaintext elements of a vector into a ciphertext and then parallelizes the homomorphic matrix-vector multiplication. However,  similar  to \cite{jiang2018secure},   AP focuses on the SIMD operation between two packed ciphertexts.

Several works design solutions for the scenario of secure inference \cite{juvekar2018gazelle,chen2020maliciously,zhang2021gala}, mostly built on the \textit{semi-honest adversary model}. Among them, one of the most  remarkable works  is GAZELLE \cite{juvekar2018gazelle}.  It is customized for the HE-GC based secure inference, and has been integrated into some advanced solutions such as DELPHI \cite{mishra2020delphi}  and EzPC \cite{chandran2019ezpc}. In fact,  the core idea of GAZELLE is also used in the design of MUSE and SIMC, which inherit DELPHI's optimization strategy for HE-based linear operations. GAZELLE builds a new homomorphic linear algebra kernel, which provides fast algorithms for mapping neural network layers to optimized homomorphic matrix-vector multiplication and convolution routines. However, the computation complexity of GAZELLE is still non-negligible. For example,  to calculate the product of a $n_o\times n_i$ plaintext matrix and a  $n_i$-dimensional ciphertext vector, GAZELLE requires at least $(\log_2(\frac{n}{n_o})+\frac{n_in_o}{n}-1)$ rotation operations, where $n$ is the number of slots in the ciphertext. This is very computationally expensive compared to other homomorphic operations such as addition and multiplication. To alleviate this problem,  Zhang \textit{et al.} \cite{zhang2021gala} present GALA, an optimized solution over GAZELLE, to reduce the complexity of the rotation operations required by matrix vector calculations from $(\log_2(\frac{n}{n_o})+\frac{n_in_o}{n}-1)$ to $(\frac{n_in_o}{n}-1)$, thus substantially improving the efficiency. However, GALA is specially customized for the secure inference under the \textit{semi-honest adversary model}. In addition, we will demonstrate that the computation complexity of GALA is not optimal and can be further optimized.

To construct an efficient GC-based protocol, the following problems generally need to be solved: (1) avoiding using the GC to perform the multiplication between elements as much as possible, which usually requires at least $O(\kappa^{2}\lambda)$ communication complexity \cite{demmler2015aby}; (2) ensuring the correctness of the client's input, so that the output of the GC is trusted and can be correctly propagated to the subsequent layers. SIMC presents  a novel protocol that can meet these two requirements.  Specifically, instead of using the GC directly to calculate $s=\alpha t$ and authenticated shares of $v=f(t)$, SIMC only resorts to the GC to obtain the garbled labels of the bits corresponding to each function, namely $s[i], v[i]$ and $\alpha v[i]$ for $1\leq i\leq \kappa$. Such  types of calculation is natural for the GC because it operates on a Boolean circuit. Moreover,  it avoids performing a large number of multiplication operations in the GC to achieve the binary-to-decimal conversion. Then, SIMC designs a lightweight input consistency verification method, which is used to force the client to feed the correct sharing of GC's input.  Compared with MUSE, SIMC reduces the communication overhead of each ReLU function from $2d\lambda+190\kappa\lambda+232\kappa^2$ to $2c\lambda+4\kappa\lambda+6\kappa^2$, and accelerates the calculation by several times.

We point out that it is possible to further simplify the  protocol in SIMC. It stems from the insight that the ReLU function can be parsed as $f(t)=t\cdot sign(t)$, where the sign function $sign(t)$ equals 1 if $t\geq 0$ and 0 otherwise. Therefore,  it is desirable if only the non-linear part of $f(t)$ (\textit{i.e.}, $sign(t)$) is encapsulated into the GC,  \textit{i.e.}, for $1\leq i\leq \kappa$, the output of the GC is $sign[i]$ and $s[i]$,  instead of $s[i]$  and $v[i]$. Then,  if we can find a lightweight alternative sub-protocol to privately compute the authenticated shares $v[i]=sign[i]\times s[i]$, it not only simplifies the size of the GC, but also further reduces the number of expensive multiplications between elements in the original GC. However,  it is challenging to build such a protocol: the replacement sub-protocol should be lightweight compared to the original GC. In addition, it should be compatible with the input consistency verification method in SIMC, so as to realize the verifiability of the client¡¯s input.

\subsection{Our Contributions}
\label{Our Contributions}
In this work, we present SIMC 2.0, a new secure inference model resilient to malicious clients and achieves up to $5\times$ performance improvement over the previous state-of-the-art SIMC.  SIMC 2.0 complies with the underlying structure of SIMC, but designs highly optimized methods to substantially reduce the overhead of linear and non-linear layers. In short, the contributions of SIMC 2.0 are summarized as follows:
\begin{itemize}[leftmargin=*,align=left]
\item  We design a new coding method for homomorphic linear computation in an SIMD manner. It is custom-built through the insight into the complementarity between cryptographic primitives in an HE-GC based framework, where the property of secret sharing is used to convert a large number of private rotation operations to be executed in the plaintext environment. Moreover, we design a block-combined diagonal encoding method  to further reduce the number of rotation operations.  As a result, compared with SIMC, we reduce the complexity of rotation operations required by the matrix-vector computations from $(\log_2(\frac{n}{n_o})+\frac{n_in_o}{n}-1)$ to $(l-1)$, where $l$ is a hyperparameter set by the server. We also present a new method for the homomorphic convolution operation with the SIMD support.
\item We reduce the size of the GC in SIMC by about two thirds. As discussed above, instead of using the GC to calculate the entire ReLU function, we only encapsulate the non-linear part of ReLU into the GC. Then, we construct a lightweight alternative protocol that takes the output of the  simplified GC as input to calculate the sharing of the desired result. We exploit the authenticated shares \cite{wang2017authenticated} as the basis to build the lightweight alternative protocol. As a result,  compared with SIMC, SIMC 2.0 reduces the communication overhead of calculating each ReLU  from $2c\lambda+4\kappa\lambda+6\kappa^2$ to $2e\lambda+4\kappa\lambda+6\kappa^2+2\kappa$,    where $e<c$ denotes the number of AND gates required in the GC.
 \item We demonstrate that SIMC 2.0 has the same security properties as SIMC, \textit{i.e.}, it is secure against the malicious client model. We prove this with theoretical analysis following the similar logic of SIMC.  We use several datasets (\textit{e.g.}, MNIST, CIFAR-10) to conduct extensive experiments on various ML models.
 Our experiments show that SIMC 2.0  achieves a significant speedup by up to $17.4\times $ for linear layer computation and at least  $1.3\times$ communication reduction in  non-linear layer computation under different data dimensions. Meanwhile, SIMC 2.0 demonstrates an encouraging runtime boost by $2.3\sim 4.3\times$ over SIMC on different state-of-the-art ML models (\textit{e.g.}, AlexNet, VGG, ResNet).

\end{itemize}

\noindent\textbf{Roadmap}: The remainder of this paper is organized as follows. In  Section \ref{sec:PROBLEM STATEMENT}, we review some basic concepts and introduce the scenarios and threat models involved in this paper. In  Sections \ref{subsub:Linear layer optimization} -- \ref{Secure Inference}, we give the details of our  SIMC 2.0.  Performance evaluation  is presented in Section \ref{sec:PERFORMANCE EVALUATION}, and Section \ref{sec:conclusion} concludes the paper.

\section{Preliminaries}
\label{sec:PROBLEM STATEMENT}

\subsection{Threat Model}
We consider a two-party ML inference scenario consisting of a server $S_0$ and a client $S_1$. $S_0$ holds an ML model $M$ with the private and sensitive weight $W$. It is considered semi-honest: it obeys the deployment procedure of ML inference, but may be curious to infer the client's data by observing the data flow in the running process. $S_1$ holds the private input $t$. It is malicious and can arbitrarily violate the protocol specification. The model architecture $\mathrm{NN}$ of $M$ is assumed to be known by both the server and the client. Our goal is to design a secure inference framework, which enables $S_1$ to learn the inference result of $t$ without knowing any details about the model weight $W$, and $S_0$ knows nothing about the input $t$. A formal definition of the threat model is provided in Appendix~\ref{A:threat model}.

\subsection{Notations}
We describe some notations used in this paper. Specifically,  $\lambda$ and $\sigma$ denote the computational and statistical security parameters.  For any $n>0$, $[n]$ denotes the set $\{1, 2, \cdots n\}$. In our SIMC 2.0, all arithmetic operations such as addition and multiplication are performed in the field $\mathbb{F}_p$, where $p$ is a prime and $\kappa=\lceil \log p \rceil$. We assume that any element $x\in \mathbb{F}_p$ can be naturally mapped to the set $\{1, 2, \cdots \kappa\}$, where we use $x[i]$  to represent the $i$-th bit of $x$, \textit{i.e.}, $x=\sum_{i\in[\kappa]}x[i]\cdot2^{i-1}$. For a vector $\mathbf{x}$ and an element $\alpha \in \mathbb{F}_p$, $\alpha+ \mathbf{x}$  and $\alpha \mathbf{x}$ denote the addition and multiplication of each component in $\mathbf{x}$ with $\alpha$, respectively.  Given a function $f: \mathbb{F}_p \rightarrow \mathbb{F}_p$, $f(\mathbf{x})$ represents the evaluation on $f$  for each component in $\mathbf{x}$. Given two elements $x, y\in \mathbb{F}_p$, $x||y$ denotes the concatenation of $x$ and $y$.

For simplicity, we assume the targeted ML architecture $\mathrm{NN}$ consists of alternating linear and nonlinear layers. Let the specifications of the linear layers and non-linear layers be $\mathbf{N}_1$, $\mathbf{N}_2$, $\cdots$, $\mathbf{N}_m$ and $f_1$, $f_2$, $\cdots$, $f_{m-1}$, respectively. Given an input vector $\mathbf{t}_0$, the model sequentially computes $\mathbf{u}_i=\mathbf{N}_i\cdot \mathbf{t}_{i-1}$ and $\mathbf{t}_i= f_i(\mathbf{u}_i)$, where $i\in[m-1]$. As a result, we have $\mathbf{u}_m=\mathbf{N}_m\cdot\mathbf{t}_{m-1}=\mathrm{NN}(\mathbf{t}_0)$.

\subsection{Fully Homomorphic Encryption (FHE)}
FHE is a public key encryption system that supports the evaluation of any function parsed as a polynomial in ciphertext. Informally, assuming that the message space is $\mathbb{F}_p$, FHE contains the following four algorithms:
\begin{itemize}[leftmargin=*,align=left]

\item $\mathtt{KeyGen}(1^\lambda)\rightarrow (pk, sk)$. Given the security parameter $\lambda$, $\mathtt{KeyGen}$ is a randomized algorithm that generates a public key $pk$ and the corresponding secret key $sk$.
\item $\mathtt{Enc}(pk, t)\rightarrow c$. Given the public key $pk$ and a message $t\in \mathbb{F}_p$ as input, $\mathtt{Enc}$ outputs a ciphertext $c$.
\item $\mathtt{Dec}(sk, c)\rightarrow t$. Given the secret key $sk$ and a ciphertext $c$ as input, $\mathtt{Dec}$ decrypts $c$ and obtains the plaintext $t$.
\item $\mathtt{Eval}(pk, c_1, c_2, F)\rightarrow t$. Given the public key $pk$ and two ciphertexts $c_1, c_2$ encrypting $t_1, t_2$, respectively, and a function $F$ parsed as a polynomial,  $\mathtt{Eval}$ outputs a ciphertext $c'$ encrypting $F(t_1, t_2)$.
\end{itemize}
We require FHE to satisfy correctness, semantic security and additive homomorphism along with function privacy\footnote{Informally, it means given a ciphertext $c$ encrypting a share of ciphertext $F(c_1, c_2)$ obtained by homomorphically evaluating $F$, $c$ is indistinguishable from another ciphertext $c'$ which is a share of ciphertext $F'(c_1 c_2)$ for any other $F'$, even given $sk$.}. In addition, we use the ciphertext packing technology (CPT) \cite{halevi2014algorithms} to accelerate the parallelism of homomorphic computation. In brief, CPT is capable of packing up to $n$ plaintexts into one ciphertext containing $n$ plaintext slots, thereby improving the parallelism of computation. This makes homomorphic addition and multiplication for the ciphertext equivalent to performing the same operation on every plaintext slot at once. For example, given two ciphertexts $c_1$ and $c_2$ encrypting the plaintext vectors $\mathbf{t}=(t_0, t_1, \cdots, t_n)$ and $\mathbf{t}'=(t'_0, t'_1, \cdots, t'_n)$ respectively,  $\mathtt{Eval}(pk, c_1,c_2, F=(a+b))$ outputs a ciphertext $c$ encrypting the plaintext vector $\mathbf{t''}=(t_0+t'_0, t_1+t'_1, \cdots, t_n+t'_n)$.

 CPT also provides a \textit{rotation} function $\mathbf{Rot}$ to facilitate rotate operations between plaintexts in different plaintext slots. To be precise, given a ciphertext $c$ encrypting a plaintext vector $\mathbf{t}=(t_0, t_1, \cdots, t_n)$, $\mathbf{Rot} (pk, c, j)$ transforms $c$ into an encryption of $\mathbf{t'}=(t_j, t_{j+1}, \cdots, t_0, \cdots, t_{j-1})$, which enables the packed vector in the appropriate position to realize the correct element-wise addition and multiplication. Since the computation cost of the rotation operation in FHE is much more expensive than other operations such as homomorphic addition and multiplication, in our SIMC 2.0, we aim to design homomorphic parallel computation methods customized for matrix-vector multiplication and convolution, thereby minimizing the number of rotation operations incurred during execution.


\subsection{Secret Sharing}
\label{Secret Sharing}
We describe some terms used for secret sharing as follows.
\begin{itemize}[leftmargin=*,align=left]

\item \textbf{Additive secret sharing}. Given $x\in \mathbb{F}_p$, we say a 2-of-2 additive
secret sharing of $x$ is a pair $	(\left \langle x \right \rangle_0, \left \langle x \right \rangle_1)=(x-r, r)\in \mathbb{F}_{p}^{2}$ satisfying $ \langle x\rangle_0+\langle x\rangle_1=x$, where $r$ is randomly selected from $\mathbb{F}_p$. Additive secret sharing is perfectly hidden, \textit{i.e.}, given $ \langle x\rangle_0$ or $\langle x\rangle_1$, the value of $x$ is perfectly hidden.
\item \textbf{Authenticated shares}.  Given a randomly selected $\alpha\in \mathbb{F}_p$ (known as the MAC key), the authenticated share of $x \in \mathbb{F}_p$ is denoted as $\{\langle x\rangle_b, \langle \alpha x\rangle_b\}_{b\in \{0, 1\}}$, where each party $S_b$ holds $(\langle x\rangle_b, \langle \alpha x\rangle_b)$. In the fully malicious protocol, $\alpha$ should be shared secretly with all parties. In our client-malicious model, consistent with previous works \cite{lehmkuhl2021muse,chandran2021simc}, $\alpha$ is picked uniformly by the server $S_0$ and secretly shared with the client $S_1$. Authenticated shares provide $\lfloor \log p \rfloor$ bits of statistical security. Specifically, assuming the malicious $S_1$ has tampered with the share of $x$ as $x+\beta$ by changing the shares $(\langle x\rangle_1, \langle \alpha x\rangle_1)$ to $(\langle x \rangle_1+\beta, \langle \alpha x\rangle_1+\beta')$, the probability of parties being authenticated to hold the share of $x+\beta$ (\textit{i.e.}, $\alpha x+\beta'=\alpha(x+\beta)$) is at most $2^{-\lfloor \log p \rfloor}$.

\item \textbf{Authenticated Beaver's multiplicative triples}. Given a randomly selected triple $(A, B, C)\in \mathbb{F}_{p}^{3}$ satisfying $AB=C$, an authenticated Beaver's multiplicative triple denotes that $S_b$ holds the following shares $$\{(\langle A\rangle_b, \langle \alpha A\rangle_b), (\langle B\rangle_b, \langle \alpha B\rangle_b), (\langle C\rangle_b, \langle \alpha C\rangle_b)\}$$ for $b\in \{0, 1\}$. The process of generating triples is offline and is used to facilitate multiplication between authenticated shares. We give details of generating such triples in Figure~\ref{Algorithm of generating authenticated Beaver's multiplicative triple} in Appendix~\ref{A:Authenticated Beaver's multiplicative triples}.
\end{itemize}

\subsection{Oblivious Transfer}
\label{Oblivious Transfer}
The 1-out-of-2 Oblivious Transfer (OT) \cite{keller2015actively} is denoted as OT$_{n}$, where the inputs of the sender ($S_0$) are two strings $s_0, s_1 \in\{0, 1 \}^{n}$, and the input of the receiver ($S_1$) is a choice bit $b\in \{0, 1\}$. At the end of the OT-execution, $S_1$ obtains $s_b$ while $S_1$ receives  nothing.  Succinctly, the security properties of OT$_{n}$ require that 1) the receiver learns nothing but $s_b$ and 2) the sender knows nothing about the choice $b$.  In this paper, we require that the instance of OT$_{n}$ is secure against a semi-honest sender and a malicious receiver.  We use OT$_{n}^{\kappa}$ to represent $\kappa$ instances of OT$_{n}$. We exploit \cite{keller2015actively} to implement OT$_{n}^{\kappa}$  with the communication complexity of $\kappa{\lambda+2n}$ bits.

\subsection{Garbled Circuits}
\label{Garbled Circuits}
A garbled scheme \cite{ghodsi2021circa} for Boolean circuits consists of two algorithms, $\mathtt{Garble}$ and $\mathtt{GCEval}$, defined as follows.
\begin{itemize}[leftmargin=*,align=left]

\item $\mathtt{Garble}(1^\lambda, C)\rightarrow (\mathtt{GC}, \{ \{\mathtt{lab}_{i,j}^{in}\}_{i\in[n]},\{\mathtt{lab}_{j}^{out}\}\}_{j\in\{0,1\}})$. Given the security parameter $\lambda$ and a Boolean circuit $C: \{0,1 \}^{n}\rightarrow \{0, 1\}$, the $\mathtt{Garble}$ function outputs a garbled  circuit $\mathtt{GC}$, an input set  $\{\mathtt{lab}_{i,j}^{in}\}_{i\in[n], j\in\{0,1\}}$ of labels  and an output set $\{\mathtt{lab}_{j}^{out}\}_{j\in\{0,1\}}$, where each label is of $\lambda$ bits. In brief, $\{\mathtt{lab}_{i,x[i]}^{in}\}_{i\in[n]}$ represents the garbled input for any $x\in \{0,1\}^{n}$ and the label $\mathtt{lab}_{ C(x)}^{out}$ represents the  garbled output for $C(x)$.
\item $\mathtt{GCEval}(\mathtt{GC}, \{\mathtt{lab}_{i}\}_{i\in[n]})\rightarrow \mathtt{lab'}$. Given the garbled circuit $\mathtt{GC}$ and a set of labels $\{\mathtt{lab}_{i}\}_{i\in[n]}$, $\mathtt{GCEval}$ outputs a label $\mathtt{lab'}$.
\end{itemize}

\begin{figure*}[htb]
\centering
\includegraphics[width=0.85\textwidth]{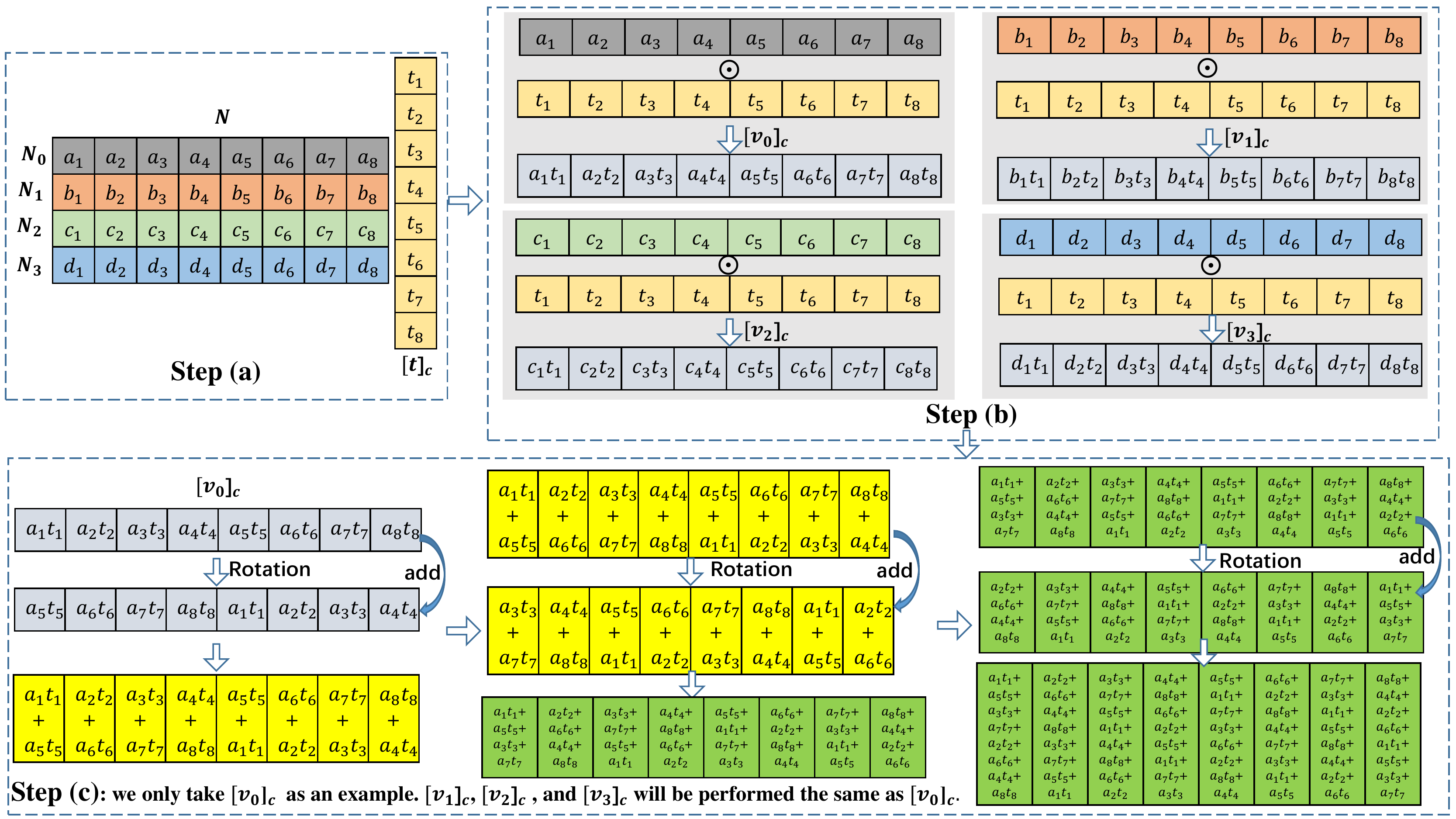}
\caption{Naive matrix-vector multiplication}
\label{Fig:Naive matrix-vector multiplication}
\end{figure*}

The above garbled scheme ($\mathtt{Garble}$, $\mathtt{GCEval}$) is required to satisfy the following properties:
\begin{itemize}[leftmargin=*,align=left]

\item \textbf{Correctness}. $\mathtt{GCEval}$ is faithfully evaluated on $\mathtt{GC}$ and outputs $C(x)$ if the garbled $x$ is given. Formally, for any circuit $C$ and $x\in \{0,1\}^{n}$, we have $$\mathtt{GCEval}(\mathtt{GC}, \{\mathtt{lab}_{i, x[i]}^{in}\}_{i\in[n]})\rightarrow \mathtt{lab}_{ C(x)}^{out}$$
\item \textbf{Security}. For any circuit $C$ and input $x\in \{0,1\}^{n}$, there exists a polynomial probability-time simulator $\mathtt{Sim}$ that can simulate the $\mathtt{GC}$ and garbled input of $x$  generated by $\mathtt{Garble}$ under real execution, \textit{i.e.}, $(\mathtt{GC},\{\mathtt{lab}_{i, x[i]}^{in}\}_{i\in[n]})\approx \mathtt{Sim}(1^\lambda, C)$, where $\approx$ indicates computational indistinguishability
of   two distributions $(\mathtt{GC},\{\mathtt{lab}_{i, x[i]}^{in}\}_{i\in[n]})$ and $\mathtt{Sim}(1^\lambda, C)$.
\item \textbf{Authenticity}. It is infeasible to guess the output label of $1-C(x)$ given the garbled input of $x$ and $\mathtt{GC}$. Formally, for any circuit $C$ and $x\in \{0,1\}^{n}$, we have $\left( \mathtt{lab}_{ 1-C(x)}^{out}|\mathtt{GC},\{\mathtt{lab}_{i, x[i]}^{in}\}_{i\in[n]}\right)\approx U_{\lambda}$, where $U_{\lambda}$ represents the uniform distribution on the set $\{0,1\}^{n}$.
\end{itemize}
Note that the garbled scheme described above can be naturally extended to the case with multiple garbled outputs. We also utilize state-of-the-art optimization strategies, including point-and-permute,
free-XOR and half-gates \cite{zahur2015two} to construct the garbled circuit. We provide a high-level description of performing a secure two-party computation with $\mathtt{GC}$: assuming that a semi-honest server $S_0$ and a malicious client $S_1$ hold private inputs $x$ and $y$, respectively. Both parties evaluate $C(x,y)$  on $\mathtt{GC}$ as follows. $S_0$ first garbles the circuit $C$  to learn  $\mathtt{GC}$  and garbled labels about the input and output of  $\mathtt{GC}$. Then,
both parties  invoke the OT functionality, where the sender $S_0$ inputs garbled labels corresponding to the input
wires of $S_1$, while the receiver $S_1$ inputs $y$ to obtain garbled inputs of $y$. $S_0$ additionally sends the garbled input of $x$, $\mathtt{GC}$  and a pair of ciphertexts for every output wire $w$ of $C$ to $S_1$. As a result,  $S_1$ evaluates $\mathtt{GC}$ to learn the garbled output of $C(x,y)$ with the received garbled input about $x$ and $y$.  From the garbled output and the
pair of ciphertexts given for each output wire,  $S_1$ finally gets $C(x,y)$. $S_1$ sends $C(x, y)$ along with the hash of the garbled output to $S_0$ for verification.  $S_0$ accepts it if the hash value  corresponds to $C(x, y)$.

  \section{Linear layer optimization}
\label{subsub:Linear layer optimization}
We start by describing the secure execution of linear layers in SIMC.  To be precise,  SIMC designs two protocols (\textbf{InitLin} and \textbf{Lin}) to securely perform linear layer operations, as shown in \textbf{Figures}~\ref{Protocols InitLin} and \ref{Protocol Lin} in Appendix~\ref{A:Protocols InitLin and Lin}.  We observe that the most computationally intensive operations in \textbf{InitLin} and \textbf{Lin} are homomorphically computing $\mathbf{N\cdot t}$ and $\alpha\mathbf{N\cdot t}$ (including matrix-vector multiplication in FC layers and convolution operations in convolution layers), where $\mathbf{N}$ is a plaintext weight matrix held by the server, and $\mathbf{t}$ (will be encrypted as $\mathbf{[t]_c}$) is the input held by the client.

We focus on the optimization of matrix-vector parallel multiplication.  More specifically, we consider an FC layer with $n_i$ inputs and $n_o$ outputs, \textit{i.e.}, $\mathbf{N}\in \mathbb{F}_p^{n_o\times n_i}$. The client's input is a ciphertext vector $\mathbf{t}\in \mathbb{F}_p^ {n_i}$.   $n$ is the number of slots in a ciphertext. We first describe a naive approach to parallelize homomorphic multiplication of $\mathbf{N\cdot t}$, followed by a state-of-the-art method proposed by GAZELLE \cite{juvekar2018gazelle}. Finally we introduce our proposed scheme that substantially reduces the computation cost of matrix-vector multiplication.

\subsection{Naive Method}
\label{subsub:Naive Method}

The naive method of matrix-vector multiplication is shown in \textbf{Figure}~\ref{Fig:Naive matrix-vector multiplication}, where $\mathbf{N}$ is the $n_o\times n_i$ dimensional plaintext matrix held by the server and $\mathbf{[t]_c}$ is the encryption of the client's input vector. The server encodes each row of the matrix into a separate plaintext vector (Step (a) in \textbf{Figure}~\ref{Fig:Naive matrix-vector multiplication}), where each encoded vector is of length $n$ (including zero-padding if necessary).  We denote these encoded plaintext vectors as $\mathbf{N_0}$, $\mathbf{N_1}$, $\cdots$, $\mathbf{N_{n_o-1}}$. For example, there are four vectors in  \textbf{Figure}~\ref{Fig:Naive matrix-vector multiplication}, namely $\mathbf{N_0}$, $\mathbf{N_1}$, $\mathbf{N_2}$, and $\mathbf{N_3}$.

The purpose of the server is to homomorphically compute the dot product of the plaintext $\mathbf{N}$ and the ciphertext $\mathbf{[t]_c}$. Let ScMult be the scalar multiplication of a plaintext and a ciphertext in HE. The server first uses ScMult to compute the element-wise multiplication of $\mathbf{N_i}$ and $\mathbf{[t]_c}$. As a result, the server gets $\mathbf{[v_i]_c}=\mathbf{[N_i\odot t]_c}$ (Step (b) in \textbf{Figure}~\ref{Fig:Naive matrix-vector multiplication}). We observe that the sum of all elements in $\mathbf{[v_i]_c}$ is the $i$-th element of the desired dot product of $\mathbf{N}$ and $\mathbf{[t]_c}$. Since there is no direct way to obtain this sum in HE, we rely on the rotation operation  to do it. To be precise, as shown in Step (c) in \textbf{Figure}~\ref{Fig:Naive matrix-vector multiplication}, $\mathbf{[v_i]_c}$ is first rotated by permuting $\frac{n_i}{2}$ positions such that the first $\frac{n_i}{2}$ entries of the rotated $\mathbf{[v_i]_c}$ are the second $\frac{n_i}{2}$ entries of the original $\mathbf{[v_i]_c}$. The server then performs element-wise addition on the original $\mathbf{[v_i]_c}$ and the rotated one homomorphically, which derives a ciphertext whose sum of the first $\frac{n_i}{2}$ entries is actually the desired result. The server iteratively performs the above rotation operation for $\log_2{n_i}$ times. Each time it operates on the ciphertext derived from the previous iteration. The server finally gets the ciphertext whose first entry is the $i$-th element of $\mathbf{Nt}$. By executing this procedure for each $\mathbf{[v_i]_c}$, \textit{i.e.}, $\mathbf{[v_0]_c}$, $\mathbf{[v_1]_c}$, $\mathbf{[v_2]_c}$, and $\mathbf{[v_3]_c}$ in \textbf{Figure}~\ref{Fig:Naive matrix-vector multiplication}, the server obtains $n_o$ ciphertexts. Consistently, the first entries of these ciphertexts correspond to the elements in $\mathbf{Nt}$.

\begin{figure*}[htb]
\centering
\includegraphics[width=0.85\textwidth]{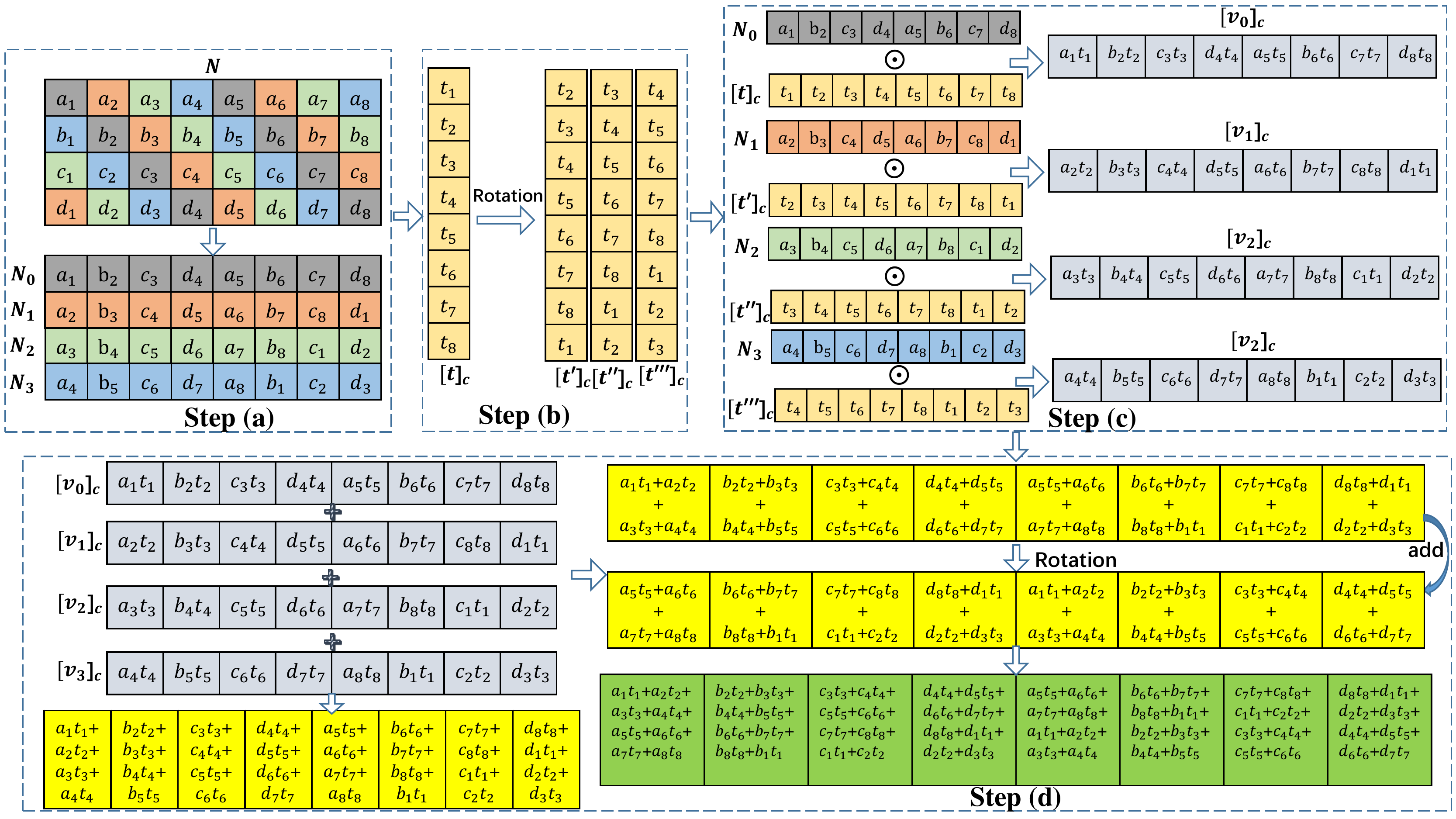}
\caption{Hybrid matrix-vector multiplication}
\label{Fig:Hybrid matrix-vector multiplication}
\end{figure*}

We now analyze the complexity of the above matrix-vector multiplication. We consider the process starting with the server receiving the ciphertext (\textit{i.e.}, $\mathbf{[t]_c}$) from the client until it obtains the ciphertext (\textit{i.e.}, $n_o$ ciphertexts) to be shared\footnote{In the neural network inference process with HE-GC as the underlying structure, the ciphertext generated in the linear phase will be securely shared between the client and the server, and used as the input of the subsequent GC-based nonlinear function. Readers can refer to \cite{lehmkuhl2021muse,chandran2021simc,patrablaze} for more details.}. There are totally $n_o$ scalar multiplication operations (\textit{i.e.}, ScMult), $n_o\log_2{n_i}$ rotation operations and $n_o\log_2{n_i}$ addition operations. It produces $n_o$ output ciphertexts, each containing an element in $\mathbf{Nt}$. Such naive use of the ciphertext space inevitably incurs inefficiencies for linear computations.

\subsection{Hybrid Method (GAZELLE)}
\label{subsub:Hybrid Method}
In order to fully utilize $n$ slots and substantially reduce the computation complexity, Juvekar \textit{et al.} \cite{juvekar2018gazelle} propose GAZELLE, a state-of-the-art hybrid method customized for the HE-GC based secure inference. It has been integrated into some advanced solutions such as DELPHI \cite{mishra2020delphi}  and EzPC\cite{chandran2019ezpc}. The core idea of GAZELLE also inspires the design of MUSE and SIMC, as they inherit DELPHI's optimization strategy for HE-based linear operations. GAZELLE  is actually a variant of diagonal encoding \cite{halevi2014algorithms}, which exploits the fact that $n_o$ is usually much smaller than $n_i$ in the FC layer. Based on this, GAZELLE shows that the most expensive rotation operation is a function of $n_o$ rather than $n_i$, thus speeding up the calculation of the FC layer.

The basic idea of GAZELLE is shown in \textbf{Figure}~\ref{Fig:Hybrid matrix-vector multiplication}. The server encodes the matrix $\mathbf{N}$ diagonally into a vector of $n_o$ plaintexts. For example, as shown in Step (a) in \textbf{Figure}~\ref{Fig:Hybrid matrix-vector multiplication}, the first plaintext vector $\mathbf{N_0}$ consists of gray elements in matrix $\mathbf{N}$, \textit{i.e.}, ($a_1, b_2, c_3, d_4, a_5, b_6, c_7, d_8$). The second  plaintext vector $\mathbf{N_1}$ consists of orange elements ($a_2$, $b_3$, $c_4$, $d_5$, $a_6$, $b_7$, $c_8$, $d_1$). The composition of $\mathbf{N_2}$ and $\mathbf{N_3}$ is analogous. Note that the meaning of $\mathbf{N_0}$ to $\mathbf{N_3}$ in the hybrid approach is different from the previous naive approach in Section \ref{subsub:Naive Method}.

For $i=1$ to $n_o-1$, the server rotates $\mathbf{[t]_c}$ by $i$ positions, as shown in Step (b) in \textbf{Figure}~\ref{Fig:Hybrid matrix-vector multiplication}. Afterwards, ScMult is used to perform the element-wise multiplication of $\mathbf{N_i}$ and the corresponding ciphertext. For example, as shown in Step (c) in \textbf{Figure}~\ref{Fig:Hybrid matrix-vector multiplication}, $\mathbf{N_0}$ is multiplied by the ciphertext $\mathbf{[t]_c}$, while $\mathbf{N_1}$ is multiplied by $\mathbf{[t']_c}$, and so on.  As a result, the server obtains $n_o$ multiplied ciphertexts, $\{\mathbf{[v_i]_c}\}$.
The server receives four ciphertexts $\{\mathbf{[v_0]_c}, \mathbf{[v_1]_c}, \mathbf{[v_2]_c}, \mathbf{[v_3]_c}\}$, whose elements are all part of the desired dot product. Then, the server sums all ciphertexts $\mathbf{[v_i]_c}$ (Step (d) in \textbf{Figure}~\ref{Fig:Hybrid matrix-vector multiplication}) in an element-wise way to form a new ciphertext. At this point, the server performs rotation operations similar to the naive approach, \textit{i.e.},  iteratively performing $\log_2\frac{n_i}{n_o}$ rotations followed by addition to obtain a final single ciphertext.
The first $n_o$ entries of this ciphertext correspond to the $n_o$ ciphertext elements of $\mathbf{Nt}$.

To further reduce the computation cost, GAZELLE proposes combining multiple copies of $\mathbf{t}$ into a single ciphertext (called $\mathbf{[t_{pack}]_c}$), since the number $n$ of slots in a single ciphertext is usually larger than the dimension $n_i$ of the input vector.  As a result, $\mathbf{[t_{pack}]_c}$ has $\frac{n}{n_i}$ copies of $\mathbf{t}$ so that the server can perform ScMult operations on $\mathbf{[t_{pack}]_c}$ with $\frac{n}{n_i}$ encoded vectors at one time. This causes the server to get $\frac{n_in_o}{n}$ ciphertexts instead of $n_o$ ciphertexts, which results in a single ciphertext holding $\frac{n}{n_o}$ rather than $\frac{n_i}{n_o}$ blocks. Finally, the server iteratively performs $\log_2\frac{n}{n_o}$ rotations followed by addition to obtain a single final ciphertext, where the first $n_o$ entries in this ciphertext correspond to the $n_o$ ciphertext elements of $\mathbf{Nt}$.

\begin{figure*}[htb]
\centering
\includegraphics[width=1.0\textwidth]{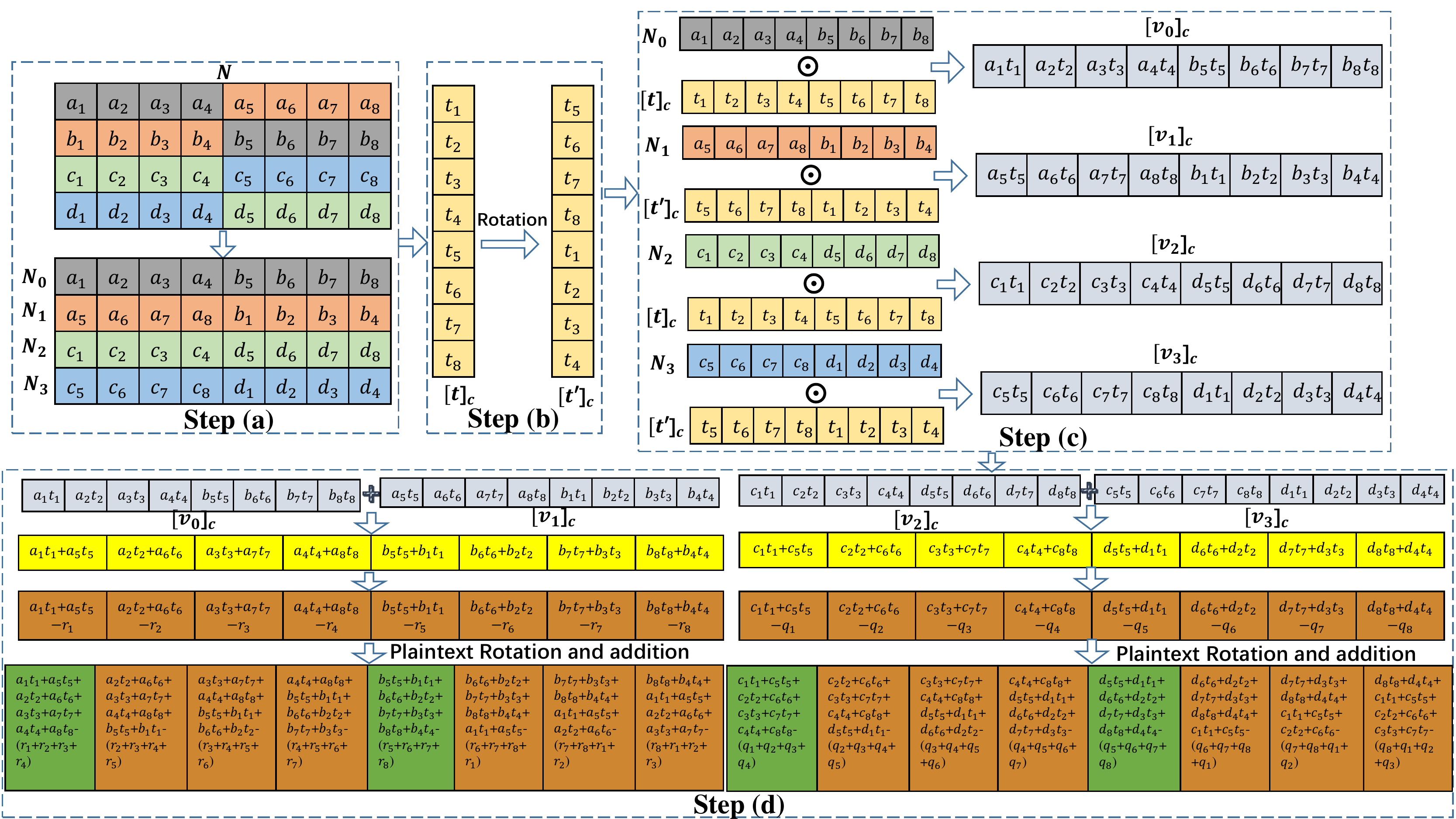}
\caption{Our matrix-vector multiplication algorithm}
\vspace{-10pt}
\label{Fig:Our matrix-vector multiplication}
\end{figure*}

In terms of complexity, GAZELLE requires $\frac{n_in_o}{n}$ scalar multiplications, $\log_2\frac{n}{n_o}+\frac{n_in_o}{n}-1$ rotations, and $\log_2\frac{n}{n_o}+\frac{n_in_o}{n}-1$ additions. It outputs only one ciphertext, which greatly improves the efficiency and utilization of slots over the naive approach.

\subsection{Our Method}
\label{Sec:our method}

In the above hybrid method, the rotation operation comes from two parts: rotations required for the client's encrypted input (Step (b) in \textbf{Figure}~\ref{Fig:Hybrid matrix-vector multiplication}) and the subsequent rotations followed by addition involved in obtaining the final matrix-vector result (Step (d) in \textbf{Figure}~\ref{Fig:Hybrid matrix-vector multiplication}). Therefore, our approach is motivated by two observations in order to substantially reduce the number of rotation operations incurred by these two aspects. First, we divide the server's original $n_o\times n_i$-dimensional weight matrix $\mathbf{N}$ into multiple sub-matrices, and calculate each sub-matrix separately (explained later based on \textbf{Figure}~\ref{Fig:Our matrix-vector multiplication}). This can effectively reduce the number of rotations for the client's ciphertext $\mathbf{[t]_c}$ from $n_o-1$ to $l-1$, where $l$ is a hyperparameter set by the server.

Second, for the rotations followed by addition (Step (d) in \textbf{Figure}~\ref{Fig:Hybrid matrix-vector multiplication}), we claim that they are actually not necessary. This is determined by the characteristics of the HE-GC based secure inference framework. To be precise, the ciphertexts output from the FC layer will be secretly shared with the client and server, and used as the input of the next non-linear layer function. As the shares are in plaintext, we can completely convert the rotation followed by addition to be performed under plaintext. This will significantly reduce the computation complexity. For example, given a $16\times256$ -dimensional matrix and  a  vector  of length 256, our experiments indicate at least a $3.5\times$ speedup compared to the hybrid method (see more details in Section~\ref{sec:PERFORMANCE EVALUATION}).

\textbf{Figure}~\ref{Fig:Our matrix-vector multiplication} illustrates our matrix-vector calculation procedure.  Specifically, we first split the matrix $\mathbf{N}$ into multiple sub-matrices. As shown in Step (a) in \textbf{Figure}~\ref{Fig:Our matrix-vector multiplication}, for simplicity, we take the first $l=2$ rows (\textit{i.e.}, $\mathbf{N}_0$ and $\mathbf{N}_1$) of $\mathbf{N}$ as a sub-matrix, thus $\mathbf{N}_2$ and $\mathbf{N}_3$  form another sub-matrix. Then,  we exploit the diagonal method to arrange the $1\times (n_i/l)$ sized sub-matrices of $\mathbf{N}$ as a sub-matrix, which results in two new sub-matrices. For subsequent ciphertext parallel computation, we need to sequentially perform $l-1$ rotations of the client's input $\mathbf{[t]_c}$, starting from moving the $(n_i/l+1)$-th entry to the first entry. In Step (b) in \textbf{Figure}~\ref{Fig:Our matrix-vector multiplication}, $\mathbf{[t]_c}$ is rotated into $\mathbf{[t']_c}$, where the first entry of $\mathbf{[t']_c}$ is the $(8/2+1)=5$-th entry of $\mathbf{[t]_c}$.

ScMult is further used to perform the element-wise multiplication of $\mathbf{N_i}$ and the corresponding ciphertext. For example, as shown in Step (c) in \textbf{Figure}~\ref{Fig:Our matrix-vector multiplication}, $\mathbf{N_0}$ is multiplied by the ciphertext $\mathbf{[t]_c}$, while $\mathbf{N_1}$ is multiplied by $\mathbf{[t']_c}$. $\mathbf{N_2}$ and $\mathbf{N_3}$ are operated similarly. As a result, the server obtains $n_o$ multiplied ciphertexts $\{\mathbf{[v_i]_c}\}$, which can be divided into $l$ independent ciphertext pairs.
We can observe that the server gets two independent ciphertext pairs, \textit{i.e.},  $\{\mathbf{[v_0]_c}, \mathbf{[v_1]_c}\}$ and $ \{\mathbf{[v_2]_c}, \mathbf{[v_3]_c}\}$. For each individual set of ciphertexts, the server adds up all ciphertexts in that set in an element-wise way,  thus forming $l$ new ciphertexts (yellow boxes in Step (d) in \textbf{Figure}~\ref{Fig:Our matrix-vector multiplication}).

Until now, a natural approach is to perform rotation followed by addition on each ciphertext, as in the hybrid approach. This will take a total of $n_i/l$ rotations and derive $n_o/l$ ciphertexts. However, we argue that such ciphertext operations are not necessary and can be converted to plaintext to be performed. Our key insight is that
the ciphertext result of the FC layer will be secretly shared with the client and server. Therefore, we require the server to generate  random vectors (\textit{i.e.}, \{$\mathbf{r_i}\}_{i\in{[8]}}$ and $\{\mathbf{q_i}\}_{i\in{[8]}}$ in Step (d) of \textbf{Figure}~\ref{Fig:Our matrix-vector multiplication}), and then subtract the random vectors from the
corresponding ciphertexts. The server returns the subtracted
ciphertexts to the client, which decrypts them and executes $n_i/l$ plaintext rotation followed by addition operations on them to obtain its share.  Similarly, the server gets its share by executing  $n_i/l$ plaintext rotation followed by addition operations on its  random vectors.
Compared to GAZELLE, our method only needs to perform $l-1$ rotations on the client's input, and does not require any rotations subsequently. This significantly improves the performance of matrix-vector computations.

We  further reduce the computation cost by packing multiple copies of $\mathbf{t}$ into a single ciphertext (called $\mathbf{[t_{pack}]_c}$).  As a result, $\mathbf{[t_{pack}]_c}$ has $\frac{n}{n_i}$ copies of $\mathbf{t}$ so that the server can perform ScMult operations on $\mathbf{[t_{pack}]_c}$ with $\frac{n}{n_i}$ encoded vectors at one time. The computation cost of our method is $\frac{n_in_o}{n}$ homomorphic multiplication operations, $l-1$ rotation operations, and $\frac{n_in_o}{n}-1$ homomorphic addition operations. It outputs $n_i\cdot n_o/(l\cdot n)$ ciphertexts.

\vspace{5pt}
\noindent\textit{Remark 3.1}. Table~\ref{Complexity of Each  Method} shows the comparison of our method with existing approaches in computation complexity. It is obvious that our method has better complexity, especially for rotation operations. Note that GALA \cite{zhang2021gala} also designs an improved version of GAZELLE \cite{juvekar2018gazelle} for matrix-vector multiplication, which has the computation complexity of $(\frac{n_in_o}{n}-1)$ for rotation operations. However, GALA is specially customized for secure inference under the \textit{semi-honest adversary model}. It is not clear whether it can be smoothly transferred to the \textit{client-malicious adversary model}. Moreover, our method still outperforms GALA in the computation complexity.
\renewcommand \thetable{\Roman{table}}
\renewcommand\tablename{TABLE}
\begin{table}[htb]
\centering
\footnotesize
\caption{Computation complexity of each method}
\label{Complexity of Each  Method}
\begin{tabular}{c|c|c|c}
\Xhline{1pt}
\textbf{Method} & \textbf{\#Rotation} & \textbf{\#ScMul} & \textbf{\#Add} \\
\Xhline{1pt}
Naive&$n_o\log_2{n_i}$& $n_o$& $n_o\log_2{n_i}$\\
\hline
Hybrid&$\log_2\frac{n}{n_o}+\frac{n_in_o}{n}-1$& $\frac{n_in_o}{n}$& $\log_2\frac{n}{n_o}+\frac{n_in_o}{n}-1$\\
\hline
Our method&$l-1$& $\frac{n_in_o}{n}$& $\frac{n_in_o}{n}-1$\\
\Xhline{1pt}
\end{tabular}
\end{table}

\vspace{5pt}
\noindent\textit{Remark 3.2}. We also describe the optimization of convolution operations for linear layers.  In brief, we assume that the server has $c_o$ plaintext kernels of size $k_w\times k_h\times c_i$, and the ciphertext input sent by the client to the server is $c_i$ kernels of size $u_w\times u_h$. The server is required to perform  homomorphic convolution operations between the ciphertext input and its own plaintext kernel to obtain the ciphertext output.  To improve the  parallel processing capabilities, GAZELLE proposes to pack the input data of $c_n$ channels into one ciphertext,  and requires a total of $\frac{c_i(k_wk_h-1)}{c_n}$ rotation operations to achieve convolution.  We present an improved solution over GAZELLE to execute convolution operations. We design strategies to significantly reduce the number of rotation operations involved in computing the convolution between each of the  $\frac{c_0c_i}{c_n^2}$ blocks and the corresponding input channels (see Appendix~\ref{Optimization for convolution operations}). As a result, our method reduces the computation complexity with respect to rotations by a factor of $\frac{c_i}{c_n}$ compared to GAZELLE. Readers can refer to Appendix~\ref{Optimization for convolution operations} for more details.

\section{Nonlinear Layer Optimization}
\label{SC:nonliner layer optimization}
In this section, we describe our proposed optimization method for non-linear layers. We focus on the secure computation of the activation function ReLU, one of the most popular functions in non-linear layers of modern DNNs. As shown in Section~\ref{subsub:Linear layer optimization}, the output of each linear layer will be securely shared with the server and client (\textbf{Figures}~\ref{Protocols InitLin} and \ref{Protocol Lin} in Appendix~\ref{A:Protocols InitLin and Lin}), and used as the input of the next non-linear layer to obtain authenticated shares of the output of the non-linear layer. Specifically, assume the output of a linear layer is $\mathbf{u}=\left \langle \mathbf{u} \right \rangle_0+\left \langle \mathbf{u} \right \rangle_1$, where $\left \langle \mathbf{u} \right \rangle_0$ and $\left \langle \mathbf{u} \right \rangle_1$ are the shares held by the server and the client, respectively. Then the functionality of the next nonlinear layer is shown in \textbf{Figure}~\ref{Functionality of the nonlinear layer}. At the high level, the main difference between our method and SIMC comes from parsing ReLU as $f(\mathbf{u})=\mathbf{u}\cdot sign(\mathbf{u})$ and encapsulating only the non-linear part ($sign(\mathbf{u})$) into the GC.  The sign function $sign(t)$ equals 1 if $t\geq 0$ and 0 otherwise. In this way, we can reduce the original GC size by about two-thirds, thereby further reducing the computation and communication costs incurred by running non-linear layers.

\renewcommand\tablename{Figure}
\renewcommand \thetable{\arabic{table}}
\setcounter{table}{3}
\begin{table}[htb]
\centering
\footnotesize
\begin{tabular}{|p{8cm}|}
\hline
Function $f: \mathbb{F}_p \rightarrow \mathbb{F}_p$.\\
 \textbf{Input:} $S_0$ holds $\left \langle \mathbf{u} \right \rangle_0 \in \mathbb{F}_{p}$  and  a MAC key $\alpha$ uniformly chosen from $\mathbb{F}_p$.  $S_1$ holds  $\left \langle \mathbf{u} \right \rangle_1 \in \mathbb{F}_{p}$.\\
 \textbf{Output:} $S_b$ obtains $\{(\langle \alpha\mathbf{u}\rangle_b, \langle f(\mathbf{u})\rangle_b, \langle \alpha f(\mathbf{u})\rangle_b)\}$ for $b\in\{0, 1\}$.\\
\hline
\end{tabular}
\caption{Functionality of the nonlinear layer}
\label{Functionality of the nonlinear layer}
\end{table}

Similar to SIMC, our method can be divided into four phases: Garbled Circuit phase, Authentication phase 1, Local Computation phase, and Authentication phase 2. Assume the server's ($S_0$) input is ($\left \langle \mathbf{u} \right \rangle_0$, $\alpha$) and the client's ($S_1$) input is ($\left \langle \mathbf{u} \right \rangle_1$). We provide a high-level view of the protocol. \textbf{Figure}~\ref{Protocol of the nonlinear layer} gives a detailed technical description.

\begin{table*}[htb]
\centering
\small
\begin{tabular}{|p{17.5cm}|}
\hline
\textbf{Preamble}: To compute the ReLU function $f:\mathbb{F}_p \rightarrow \mathbb{F}_p$, we consider a Boolean circuit $booln^f$ that takes the share of $\mathbf{u}$ as input and outputs $(\mathbf{u}, f(\mathbf{u}))$. In addition, we define a truncation function $\mathbf{Trun}_h: \{0, 1\}^\lambda\rightarrow \{0, 1\}^h$, which outputs the last $h$ bits of the input, where $\lambda$ satisfies $\lambda\geq 2\kappa$.\\
 \textbf{Input:}$S_0$ holds $\left \langle \mathbf{u} \right \rangle_0 \in \mathbb{F}_{p}$  and  a MAC key $\alpha$ uniformly chosen from $\mathbb{F}_p$.  $S_1$ holds  $\left \langle \mathbf{u} \right \rangle_1 \in \mathbb{F}_{p}$.\\
 \textbf{Output:} $S_b$ obtains $\{(\langle \alpha\mathbf{u}\rangle_b, \langle f(\mathbf{u})\rangle_b, \langle \alpha f(\mathbf{u})\rangle_b)\}$ for $b\in\{0, 1\}$.\\
\textbf{Protocol}:\\
\begin{itemize}
\item[1.] Garbled Circuit Phase:
\begin{itemize}
\item $S_0$ computes $\mathtt{Garble}(1^\lambda, booln^f)\rightarrow (\mathtt{GC}, \{ \{\mathtt{lab}_{i,j}^{in}\}_{i\in[2\kappa]},\{\mathtt{lab}_{i,j}^{out}\}_{i\in[2\kappa]}\}_{j\in\{0,1\}})$. Given the security parameter $\lambda$ and a boolean circuit $booln^f$, $\mathtt{Garble}$ outputs a garbled  circuit $\mathtt{GC}$, an input set  $\{\mathtt{lab}_{i,j}^{in}\}_{i\in[2\kappa], j\in\{0,1\}}$ of labels  and an output set $\{\mathtt{lab}_{i,j}^{out}\}_{i\in[2\kappa], j\in\{0,1\}}$, where each label is of $\lambda$-bits.
\item $S_0$ (as the sender) and $S_1$ (as the receiver) invoke  the OT$_{\lambda}^{\kappa}$ (see Section~\ref{Oblivious Transfer}), where $S_0$'s inputs are  $\{\mathtt{lab}_{i,0}^{in}, \mathtt{lab}_{i,1}^{in}\}_{i\in\{\kappa+1, \cdots, 2\kappa\}}$ while $S_1$'s input is $\left \langle \mathbf{u} \right \rangle_1$. As a result, $S_1$ obtains  $\{\mathtt{\tilde{lab}}_{i}^{in}\}_{i\in\{\kappa+1, \cdots, 2\kappa\}}$. In addition, $S_0$ sends the garbled  circuit $\mathtt{GC}$ and its garbled inputs $\{ \{\mathtt{\tilde{lab}}_{i}^{in}=\mathtt{lab}_{i, \left \langle \mathbf{u} \right \rangle_0[i]}\}_{i\in[\kappa]}$ to $S_1$.
\item Given the $\mathtt{GC}$  and the garbled inputs $\{\mathtt{\tilde{lab}}_{i}^{in}\}_{i\in[2\kappa]}$, $S_1$ computes $\mathtt{GCEval}(\mathtt{GC}, \{\mathtt{\tilde{lab}}_{i}^{in}\}_{i\in[2\kappa]})\rightarrow \{\mathtt{\tilde{lab}}_{i}^{out}\}_{i\in[2\kappa]}$.
\end{itemize}
\item[2.] Authentication Phase 1:
\begin{itemize}
\item For every $i\in[\kappa]$, $S_0$ randomly selects $\rho_{i, 0}$, $\sigma_{i, 0}$ and $\tau_{i, 0}\in \mathbb{F}_p$ and sets  $(\rho_{i, 1}, \sigma_{i, 1}, \tau_{i, 1})=(1+\rho_{i, 0}, \alpha+\sigma_{i, 0}, \alpha+\tau_{i, 0})$.
\item For every $i\in [2\kappa]$ and $j\in \{0, 1\}$, $S_0$ parses $\{\mathtt{lab}_{i,j}^{out}\}$ as $\xi_{i, j}||\zeta_{i,j}$ where $\xi_{i, j}\in\{0, 1\}$ and $\zeta_{i,j}\in\{0,1\}^{\lambda-1}$.
\item For every $i\in [\kappa]$ and $j\in \{0, 1\}$, $S_0$ sends $ct_{i, \xi_{i, j}}$ and $\hat{ct}_{i, \xi_{i+\kappa, j}}$ to $S_1$, where $ct_{i, \xi_{i, j}}=\tau_{i, j}\oplus \mathbf{Trun}_\kappa(\zeta_{i,j})$ and $\hat{ct}_{i, \xi_{i+\kappa, j}}=(\rho_{i, j}||\sigma_{i, j})\oplus \mathbf{Trun}_{\kappa}(\zeta_{i+\kappa,j})$.
\item For every $i\in[2\kappa]$, $S_1$ parses $\mathtt{\tilde{lab}}_{i}^{out}$ as $\tilde{\xi}_{i}||\tilde{\zeta}_{i}$ where $\tilde{\xi}_{i}\in\{0, 1\}$ and $\tilde{\zeta}_{i}\in\{0,1\}^{\lambda-1}$.
\item For every $i\in[\kappa]$, $S_1$ computes $c_i=ct_{i, \tilde{\xi}_{i}}\oplus \mathbf{Trun}_\kappa(\tilde{\zeta}_{i})$ and $(d_i||e_i)=\hat{ct}_{i, \tilde{\xi}_{i+\kappa}}\oplus\mathbf{Trun}_{2\kappa}(\tilde{\zeta}_{i+\kappa})$.
\end{itemize}
\item[3.] Local Computation Phase:
\begin{itemize}
\item $S_0$ outputs $\langle g_1\rangle_0=(-\sum_{i\in[\kappa]}\tau_{i, 0}2^{i-1})$, $\langle g_2\rangle_0=(-\sum_{i\in[\kappa]}\rho_{i, 0}2^{i-1})$ and $\langle g_3\rangle_0=(-\sum_{i\in[\kappa]}\sigma_{i, 0}2^{i-1})$.
\item $S_1$ outputs $\langle g_1\rangle_1=(\sum_{i\in[\kappa]}c_i2^{i-1})$, $\langle g_2\rangle_1=(\sum_{i\in[\kappa]}d_i2^{i-1})$ and $\langle g_3\rangle_1=(\sum_{i\in[\kappa]}e_i2^{i-1})$.
\end{itemize}
\item[4.] Authentication Phase 2:
\begin{itemize}
\item $S_b, b\in \{0, 1\}$ randomly select a triplet of fresh authentication shares  $\{(\langle A\rangle_b, \langle \alpha A\rangle_b), (\langle B\rangle_b, \langle \alpha B\rangle_b), (\langle C\rangle_b, \langle \alpha C\rangle_b)\}$ (see \textbf{Figure}~\ref{Algorithm of generating authenticated Beaver's multiplicative triple} for selection details), where triple $(A, B, C)\in \mathbb{F}_{p}^{3}$ satisfying $AB=C$.
\item $S_0$ interacts with $S_1$  to reveal $\Gamma= \mathbf{u}- A$ and $\Lambda= g_2- B$.
\item $S_b, , b\in \{0, 1\}$ computes the $\langle z_2 \rangle_b=\langle \mathbf{u}\cdot sign(\mathbf{u}) \rangle_b$ through  $\langle z_2 \rangle_b$=$\langle C\rangle_b+ \Gamma \cdot \langle B\rangle_b+ \Lambda \cdot\langle A\rangle_b+\Gamma \cdot\Lambda$.  Also, $\langle z_3 \rangle_b=\langle \alpha \mathbf{u}\cdot sign(\mathbf{u}) \rangle_b$ can be obtained by a similarly way.
\end{itemize}
\end{itemize}\\
\hline
\end{tabular}
\vspace{5pt}
\caption{Our protocol for the nonlinear layer $\pi^{f}_{\mathtt{Non-lin}}$}
\vspace{-20pt}
\label{Protocol of the nonlinear layer}
\end{table*}

\textbf{Garbled Circuit Phase}. SIMC uses the GC to calculate every  bit of   $\mathbf{u}$  and   $ReLU(\mathbf{u})$, instead of directly calculating these values. This is efficient, since it avoids incurring a large number of multiplication operations in the GC. Our method follows a similar logic, but instead of computing $ReLU(\mathbf{u})$, we compute every bit of  $sign(\mathbf{u})$. In this phase we can reduce the original GC size by about two-thirds. To achieve this, $S_0$ first constructs a garbled circuit for $booln^f$, where the input of the circuit is the share of $\mathbf{u}$ and its output is $(\mathbf{u}, sign(\mathbf{u}))$. $S_1$ evaluates this garbled circuit on $\left \langle \mathbf{u} \right \rangle_0$ and $\left \langle \mathbf{u} \right \rangle_1$ once the correct input labels are obtained through the OT protocol. At the end, $S_1$ learns the set of output labels for the bits of $\mathbf{u}$ and $sign(\mathbf{u})$.

\textbf{Authentication Phase 1}. In the previous phase $S_1$ obtains the garbled output labels for each bit of $\mathbf{u}$ and $sign(\mathbf{u})$, denoted as $\mathbf{u}[i]$ and $sign(\mathbf{u})[i]$, respectively.  The goal of this phase is to compute the shares of $\alpha\mathbf{u}[i]$, $sign(\mathbf{u})[i]$, and $\alpha sign(\mathbf{u})[i]$. We take $\alpha\mathbf{u}[i]$  as an example to  briefly describe the procedure. Specifically, we observe that the shares of $\alpha\mathbf{u}[i]$ are either shares 0 or  $\alpha$, depending on whether $\mathbf{u}[i]$ is 0 or 1. Also,  the output of the GC contains two output labels corresponding to each $\mathbf{u}[i]$ (each one for $\mathbf{u}[i]=0$ and 1). Therefore, we denote these labels as $\mathtt{lab}_{i,0}^{out}$ and $\mathtt{lab}_{i,1}^{out}$  for $\mathbf{u}[i]=0$ and $\mathbf{u}[i]=1$, respectively.  To calculate the shares of $\alpha\mathbf{u}[i]$, $S_0$ chooses a random number $\tau_{i}\in \mathbb{F}_p$ and converts it to $\mathtt{lab}_{i,0}^{out}$. Similarly, $\tau_{i}+\alpha$ is converted to $\mathtt{lab}_{i,1}^{out}$. $S_0$ sends the two ciphertexts to $S_1$, and sets its share of $\alpha\mathbf{u}[i]$ as $-\tau_{i}$. Since $S_1$ has obtained $\mathtt{lab}_{i, \mathbf{u}[i]}^{out}$ in the previous phase, it can obviously decrypt one of the two ciphertexts to obtain its  share for $\alpha\mathbf{u}[i]$ . Computing the share of $\alpha sign(\mathbf{u})[i]$ and $sign(\mathbf{u})[i]$ follows a similar method using the output labels for $sign(\mathbf{u})$.

\textbf{Local Computation Phase}. Given the shares of $\alpha\mathbf{u}[i]$, $sign(\mathbf{u})[i]$, and $\alpha sign(\mathbf{u})[i]$, this  phase locally calculates the shares of $\alpha\mathbf{u}$, $sign(\mathbf{u})$, and $\alpha sign(\mathbf{u})$. Taking $\alpha\mathbf{u}$ as an example, as shown in \textbf{Figure}~\ref{Protocol of the nonlinear layer}, each party locally multiplies $\alpha\mathbf{u}[i]$ and $2^{i-1}$ and then sums these results to get the  share on $\alpha\mathbf{u}$ . The shares of  $sign(\mathbf{u})$ and $\alpha sign(\mathbf{u})$ are calculated in a similar way.

\textbf{Authentication Phase 2}. This phase calculates the shares of $f(\mathbf{u})=\mathbf{u}sign(\mathbf{u})$, and $\alpha f(\mathbf{u})$. Since each party holds the authenticated shares of $\mathbf{u}$ and $sign(\mathbf{u})$, it is easy to compute the authenticated shares of $f(\mathbf{u})$ for each party (see \textbf{Figure}~\ref{Protocol of the nonlinear layer}).

\vspace{5pt}
\noindent\textit{Remark 4.1}. As described above, we reduce the size of the GC  in SIMC by about two thirds. Therefore, instead of using the GC to calculate the entire ReLU function, we only encapsulate the non-linear part of ReLU into the GC. Then, we construct a lightweight alternative protocol that takes the output of the  simplified GC as input to calculate the shares of the desired result. Therefore, compared with SIMC, SIMC 2.0 reduces the communication overhead of calculating each ReLU  from $2c\lambda+4\kappa\lambda+6\kappa^2$ to $2e\lambda+4\kappa\lambda+6\kappa^2+2\kappa$,  where $e<c$ denotes the number of AND gates required in the GC and satisfies ($c-e>\kappa$). Also, we experimentally demonstrate that this simplified GC improves the running efficiency by one third comapred to the original one.

\vspace{5pt}
\noindent\textit{Remark 4.2}. Our method can be easily extended to other non-linear functions, such as $\mathtt{Maxpool}$. We follow the same idea to compute all the shares of $\mathtt{Maxpool}$'s output, where we  construct a Boolean circuit for $\mathtt{Maxpool}$ that feeds  multiple inputs $(\mathbf{u}_1, \mathbf{u}_2, \cdots \mathbf{u}_\kappa)$, and outputs the reconstructed $(\mathbf{u}_1, \mathbf{u}_2, \cdots \mathbf{u}_\kappa)$ and $f(\mathbf{u}_1, \mathbf{u}_2, \cdots \mathbf{u}_\kappa)$.

\vspace{5pt}
\noindent\textbf{Correctness}. We analyze the correctness of our protocol in \textbf{Figure}~\ref{Protocol of the nonlinear layer} as follows. Based on the correctness of OT$_{\lambda}^{\kappa}$, the client $S_1$ holds $ \{\mathtt{\tilde{lab}}_{i}^{in}=\mathtt{lab}_{i, \left \langle \mathbf{u} \right \rangle_1[i]}\}_{i\in\{\kappa+1, \cdots 2\kappa\}}$. Since for $i\in[\kappa]$, $ \{\mathtt{\tilde{lab}}_{i}^{in}$ $=\mathtt{lab}_{i, \left \langle \mathbf{u} \right \rangle_0[i]}\}_{i\in[\kappa]}$, we can get $\mathtt{\tilde{lab}}_{i}^{out}=\mathtt{lab}_{i,\mathbf{u}[i]}^{out}$ and $\mathtt{\tilde{lab}}_{i+\kappa}^{out}=\mathtt{lab}_{i+\kappa,sign(\mathbf{u})[i]}^{out}$, with the correctness of ($\mathtt{Garble}$, $\mathtt{GCEval}$) for circuit $booln^f$. Therefore, for $i\in[k]$,
we have $\tilde{\xi}_{i}||\tilde{\zeta}_{i}=\xi_{i, \mathbf{u}[i]}||\zeta_{i, \mathbf{u}[i]}$ and
$\tilde{\xi}_{i+\kappa}||\tilde{\zeta}_{i+\kappa}=\xi_{i+\kappa, sign(\mathbf{u})[i]}||\zeta_{i+\kappa, sign(\mathbf{u})[i]}$. We also have $c_i=ct_{i, \xi_{i, \mathbf{u}[i]}}$ $\oplus\mathbf{Trun}_{\kappa}(\zeta_{i, \mathbf{u}[i]})=\tau_{i, \mathbf{u}[i]}$ and $(d_i||e_i)=\hat{ct}_{i, \xi_{i+\kappa, sign(\mathbf{u})[i]}}\oplus\mathbf{Trun}_{2\kappa}(\zeta_{i+\kappa, sign(\mathbf{u})[i]})=\rho_{i, sign(\mathbf{u})[i]}||\sigma_{i,  sign(\mathbf{u})[i]}$. Based on these, we have
\begin{itemize}[leftmargin=*,align=left]
\begin{small}
\item  $g_1=\sum_{i\in[\kappa]}(c_i-\tau_{i, 0})2^{i-1}=\sum_{i\in[\kappa]}\alpha(\mathbf{u}[i])2^{i-1}=\alpha\mathbf{u}$.
\item $g_2=\sum_{i\in[\kappa]}(d_i-\rho_{i, 0})2^{i-1}=\sum_{i\in[\kappa]}(sign(\mathbf{u})[i])2^{i-1}=sign(\mathbf{u})$.
 \item $g_3=\sum_{i\in[\kappa]}(e_i-\sigma_{i, 0})2^{i-1}=\sum_{i\in[\kappa]}\alpha (sign(\mathbf{u})[i])2^{i-1}=\alpha sign(\mathbf{u})$.
 \end{small}
\end{itemize}
Since each party  holds the authenticated shares of $\mathbf{u}$ and $sign(\mathbf{u})$, we can easily compute the shares of $f(\mathbf{u})=\mathbf{u}sign(\mathbf{u})$, and $\alpha f(\mathbf{u})$. This concludes the correctness proof.

\vspace{5pt}
\noindent\textbf{Security}. Our protocol for non-linear layers has the same security properties as SIMC, \textit{i.e.}, it  is secure against the malicious client model. We provide the following theorem and prove it in Appendix~\ref{proof of th1} following the similar logic of SIMC.
\begin{theorem}
Let ($\mathtt{Garble}$, $\mathtt{GCEval}$) be a garbling scheme with the properties defined in Section~\ref{Garbled Circuits}. Authenticated shares have the properties defined in Section~\ref{Secret Sharing}. Then our  protocol for non-linear layers is secure against any malicious adversary $\mathcal{A}$ corrupting the client $S_1$.
\end{theorem}
\begin{proof}
Please refer to Appendix~\ref{proof of th1}.
\end{proof}

\section{Secure Inference}
\label{Secure Inference}
During inference, the server $S_0$'s input is the weights of all linear layer, \textit{i.e.}, $\mathbf{N}_1, \cdots, \mathbf{N}_m$  while the input of $S_1$ is $\mathbf{t}$.  The goal of our secure inference protocol $\pi_{inf}$ is to learn $\mathrm{NN}(\mathbf{t}_0)$ for the client $S_1$.  Given a  trained neural network consisted of linear layers and non-linear layers,  SIMC 2.0 utilizes the HE-GC-based techniques as the underlying architecture to perform secure inference, where linear operations (including matrix multiplication and convolution functions) are performed with our optimized HE algorithm, and non-linear operations (including ReLU and MaxPool) are encapsulated into the GC for efficient computation.

The  whole inference process  can be  divided into two phases. (1) \textit{Evaluation} phase: we perform the computation of alternating linear and nonlinear layers with appropriate parameters. (2) \textit{Consistency checking} phase: after evaluation, the server verifies the consistency of the calculations so far. The output will be released to the client if the check passes. In Appendix~\ref{protocol of secure infernece}, we provide details of the complete inference protocol, and give a formal theoretical analysis to validate the security of the designed protocol under the malicious client threat model.  In particular, we claim that SIMC 2.0 has the following security property:
\begin{theorem}
 Our secure inference protocol $\pi_{inf}$ is secure against a semi-honest server $S_0$ and any malicious adversary $\mathcal{A}$ that corrupts the client $S_1$.
\end{theorem}
\begin{proof}
Please refer to Appendix~\ref{protocol of secure infernece}.
\end{proof}

\section{Performance Evaluation}
\label{sec:PERFORMANCE EVALUATION}
In this section, we conduct experiments to demonstrate the performance of SIMC 2.0. Since SIMC 2.0 is derived from SIMC \cite{chandran2021simc}\footnote{Code is available at \url{https://aka.ms/simc}.}, the most advanced solution for secure inference under the client-malicious model,  we take it as the baseline for comparing the computation and communication overheads. In more detail, we first describe the performance comparison of these two approaches for linear layer computation (including matrix-vector computation and convolution), and then discuss their overheads in secure computation of non-linear layers. Finally, we demonstrate the end-to-end performance advantage of SIMC 2.0 compared to SIMC for various mainstream DNN models (\textit{e.g.}, AlexNet, VGG-16, ResNet-18, ResNet-50, ResNet-100,  and ResNet-152).

\subsection{Implementation Details}
Consistent with SIMC, SIMC 2.0 uses the homomorphic encryption library provided by SEAL \cite{sealcrypto} (the maximum number $n$ of slots allowed for a single ciphertext is set to 4096) to realize the calculation of linear layers, and uses the components of the garbled circuit in the EMP toolkit \cite{wang2016emp} (with the OT protocol that resists active adversaries) to realize the execution of nonlinear layers. As a result, SIMC 2.0 provides 128 bits of computational security and 40 bits of statistical security. Like SIMC, our system is implemented on the 44-bit prime field. Other parameters, such as the configuration of the zero-knowledge proof protocol \cite{chen2020maliciously} and the random numbers used to verify the consistency of the calculation, are completely inherited from the SIMC settings (refer to \cite{chandran2021simc} for more details). Our experiments are carried out in both the LAN and WAN settings. LAN is implemented with two workstations in our lab. The client workstation has AMD EPYC 7282 1.4GHz CPUs with 32 threads on 16 cores and 32GB RAM. The server workstation has Intel(R) Xeon(R) E5-2697 v3 2.6GHz CPUs with 28 threads on 14 cores and 64GB RAM. The WAN setting is based on a connection between a local PC and an Amazon AWS server with an average bandwidth of 963Mbps and running time of around 14ms.


\subsection{Performance of Linear Layers}
\label{Cost Comparison of Linear Layers}
We compare the cost of SIMC 2.0 and SIMC in the linear layer. As mentioned in Section~\ref{Sec:our method}, given a $n_o\times n_i$-dimensional matrix and a vector of length $n_i$, the matrix-vector multiplication in SIMC only produces one ciphertext, while our method derives $n_i\cdot n_o/(l\cdot n)$ ciphertexts. The smaller $l$ is, the greater the computational advantage our scheme obtains compared to SIMC. For comparison simplicity, we set $l=n_i\cdot n_o/( n)$ to ensure that SIMC outputs only one ciphertext for matrix-vector multiplication, which makes SIMC and SIMC 2.0 have exactly the same communication overhead at the linear layer, so that we can focus on the comparison of computation cost.

\subsubsection{Matrix-Vector Multiplication}

\renewcommand\tablename{TABLE}
\renewcommand \thetable{\Roman{table}}
\setcounter{table}{1}
\setcounter{figure}{5}
\begin{table}[htb]
\centering
\small
\caption{Cost of matrix-vector multiplication}
\label{Complexity of Matrix-Vector Multiplication}
\setlength{\tabcolsep}{1.6mm}{
\resizebox{\linewidth}{!}{\begin{tabular}{c|c|c|c|c|c|c|c}
\Xhline{1pt}
\multirow{3}*{\textbf{Dimension}} &\multirow{3}*{\textbf{Metric}}&\multicolumn{2}{c|}{\textbf{\# operations}} &\multicolumn{4}{c}{\textbf{Running time (ms)} }\\ \cline{3-8}
&&\multirow{2}*{\textbf{SIMC}}&\multirow{2}*{\textbf{SMIC 2.0}} & \multicolumn{2}{c|}{\textbf{SIMC}} & \multicolumn{2}{c}{\textbf{SIMC 2.0 (Speedup)}} \\
& & & & LAN & WAN & LAN & WAN\\
\Xhline{1pt}
\multirow{3}*{$1\times 4096$}&Rotation& 12& \textbf{0}&\multirow{3}*{14.4}&\multirow{3}*{39.9}&\multirow{3}*{\makecell[c]{\textbf{1.6}\\\textbf{ (8.8$\times$)}}}&\multirow{3}*{\makecell[c]{\textbf{28.5}\\\textbf{(1.4$\times$)}}}\\
&ScMult& 1& \textbf{1}&&&&\\
&Add& 12&  \textbf{0}&&&&\\ \hline
\multirow{3}*{$2\times 2048$}&Rotation& 11& \textbf{0}&\multirow{3}*{13.4}&\multirow{3}*{38.7}&\multirow{3}*{\makecell[c]{\textbf{3.0}\\ \textbf{ (4.4$\times$)}}}&\multirow{3}*{\makecell[c]{\textbf{27.6}\\\textbf{ (1.4$\times$)}}}\\
&ScMult& 1& \textbf{1}&&&&\\
&Add& 11&  \textbf{0}&&&&\\ \hline
\multirow{3}*{$4\times 1024$}&Rotation& 10& \textbf{0}&\multirow{3}*{13.1}&\multirow{3}*{38.8}&\multirow{3}*{\makecell[c]{\textbf{2.98}\\\textbf{ (4.4$\times$)}}}&\multirow{3}*{\makecell[c]{\textbf{27.7}\\\textbf{ (1.4$\times$)}}}\\
&ScMult& 1& \textbf{1}&&&&\\
&Add& 10&  \textbf{0}&&&&\\ \hline
\multirow{3}*{$8\times 512$}&Rotation& 9& \textbf{0}&\multirow{3}*{11.3}&\multirow{3}*{36.6}&\multirow{3}*{\makecell[c]{\textbf{3.8}\\\textbf{2.97 ($\times$)}}}&\multirow{3}*{\makecell[c]{\textbf{28.1}\\\textbf{ (1.3$\times$)}}}\\
&ScMult& 1& \textbf{1}&&&&\\
&Add& 9&  \textbf{0}&&&&\\ \hline
\multirow{3}*{$16\times 256$}&Rotation& 8& \textbf{0}&\multirow{3}*{10.8}&\multirow{3}*{36.5}&\multirow{3}*{ \makecell[c]{\textbf{3.01}\\\textbf{ (3.5$\times$)}}}&\multirow{3}*{\makecell[c]{\textbf{28}\\\textbf{ (1.3$\times$)}}}\\
&ScMult& 1& \textbf{1}&&&&\\
&Add& 8&  \textbf{0}&&&&\\ \hline
\multirow{3}*{$32\times 128$}&Rotation& 7& \textbf{0}&\multirow{3}*{10.1}&\multirow{3}*{35.7}&\multirow{3}*{\makecell[c]{\textbf{3.1}\\\textbf{ (3.3$\times$)}}}&\multirow{3}*{\makecell[c]{\textbf{29.1}\\\textbf{ (1.2$\times$)}}}\\
&ScMult& 1& \textbf{1}&&&&\\
&Add& 7&  \textbf{0}&&&&\\ \Xhline{1pt}
\end{tabular}}}
\end{table}

\textbf{TABLE}~\ref{Complexity of Matrix-Vector Multiplication} provides the computation complexity of SIMC and SIMC 2.0 for matrix-vector multiplication with different matrix dimensions. We observe that SIMC 2.0 greatly reduces the number of most expensive rotation operations (zero time) in HE while SIMC requires up to 12 operations (which can be calculated using the formulas in \textbf{TABLE}~\ref{Complexity of Each  Method}). In addition, SIMC 2.0 is very friendly to other homomorphic operations, including addition and multiplication. For example, the execution process involves only 1 multiplication compared to SIMC. The running time is also provided in \textbf{TABLE}~\ref{Complexity of Matrix-Vector Multiplication}, which quantifies the entire time cost for one inference session, including client processing, server processing and network latency. We observe that in the LAN setting,  SIMC 2.0 can achieve the speedup by up to $8.8\times$  for matrix-vector multiplication of different dimensions. This is mainly due to our optimization strategy for this operation. Specifically,  we first divide the server's original $n_o\times n_i$-dimensional weight matrix into multiple submatrices and calculate each submatrix separately. It can effectively reduce the number of rotations required for the client's ciphertext from $n_o-1$ to $l-1$. Then we convert the subsequent rotation operations by summing in the ciphertext into the plaintext (see Section~\ref{Sec:our method}), thereby removing all the homomorphic rotation operations required at this stage.

We observe that the speedup of SIMC 2.0 is relatively smaller under the WAN setting, which is mainly caused by the communication latency between the local client and the cloud server. This dominates the total running time compared to the computation cost of the lightweight HE. Particularly, the running time is round 14 milliseconds in our setting, while the optimized HE in SIMC 2.0 only takes about 1 to 3 milliseconds.

\subsubsection{Convolution Computation}
The complexity and running time of the convolution operation for different input sizes and kernel sizes are provided in \textbf{TABLE}~\ref{Complexity of Convolution Computation}, where the encrypted input is a data sample of size $u_w\times u_h$ with $c_i$ channels (denoted as $u_w\times u_h @c_i$) and the server holds kernels with the size of $k_w\times k_h @c_o$ and $c_i$ channels per kernel. We observe that our method substantially reduces the number of rotation operations required to compute the convolution compared to SIMC. For example, SIMC 2.0 reduces the number of rotation operations  by up to $127\times$ compared to SIMC, given the input size of $16\times 16@2048$ and  kernel size of $1\times1@512$. This benefits from our designed kernel grouping method, which reduces the most expensive rotation operation by a factor of  $\frac{c_i}{c_n}$ (refer to Appendix~\ref{Optimization for convolution operations} for more details). The results indicate that a large speedup can be obtained if the input has more channels and a small kernel size. This is very beneficial for modern models which commonly have such characteristics.
For the running time,
it is obvious that our method significantly accelerates the inference execution. In more detail, in the LAN setting, SIMC 2.0 achieves speedups of $5.7\times$, $17.4\times$, $2.5\times$, and $1.7\times$ compared to SIMC.  In the WAN setting, SIMC 2.0 also shows similar performance superiority.
\begin{table}[htb]
\centering
\small
\caption{Cost of convolution computation}
\label{Complexity of Convolution Computation}
 \setlength{\tabcolsep}{1.0mm}{\resizebox{\linewidth}{!}{
\begin{tabular}{c|c|c|c|c|c|c|c|c}
\Xhline{1pt}
\multirow{3}*{\textbf{Input}}&\multirow{3}*{\textbf{Kernel}} &\multirow{3}*{\textbf{Metric}}&\multicolumn{2}{c|}{\textbf{\#operations}} &\multicolumn{4}{c}{\textbf{Running time (s)} }\\ \cline{4-9}
&&&\multirow{2}*{\textbf{SIMC}}&\multirow{2}*{\textbf{Ours}}&\multicolumn{2}{c|}{\textbf{SIMC}}&\multicolumn{2}{c}{\textbf{SIMC 2.0 (Speedup)}}\\
&&&&&LAN&WAN&LAN&WAN\\
\Xhline{1pt}
\multirow{3}*{\makecell[c]{$16\times16$ \\$@128$}}&\multirow{3}*{\makecell[c]{$1\times1$ \\$@128$}}&Rotation& 960& \textbf{120}&\multirow{3}*{1.3}&\multirow{3}*{1.4}&\multirow{3}*{\makecell[c]{\textbf{0.23}\\\textbf{(5.7$\times$)}}}&\multirow{3}*{\makecell[c]{\textbf{0.35}\\\textbf{(4.0$\times$)}}}\\
&&ScMult& 1024& \textbf{1024}&&&&\\
&&Add& 1016&  \textbf{1016}&&&&\\ \Xhline{1pt}
\multirow{3}*{\makecell[c]{$16\times16$ \\$@2048$}}&\multirow{3}*{\makecell[c]{$1\times1$ \\$@512$}}&Rotation& 61140& \textbf{480}&\multirow{3}*{79.9}&\multirow{3}*{81.0}&\multirow{3}*{\makecell[c]{\textbf{4.60}\\\textbf{(17.4$\times$)}}}&\multirow{3}*{\makecell[c]{\textbf{5.13}\\\textbf{(15.8$\times$)}}}\\
&&ScMult& 65536& \textbf{65536}&&&&\\
&&Add& 65504&  \textbf{65504}&&&&\\ \Xhline{1pt}
\multirow{3}*{\makecell[c]{$16\times16$ \\$@128$}}&\multirow{3}*{\makecell[c]{$3\times3$ \\$@128$}}&Rotation& 1024& \textbf{184}&\multirow{3}*{2.6}&\multirow{3}*{2.8}&\multirow{3}*{\makecell[c]{\textbf{1.04}\\\textbf{(2.5$\times$)}}}&\multirow{3}*{\makecell[c]{\textbf{1.22}\\\textbf{(2.3$\times$)}}}\\
&&ScMult& 9216& \textbf{9216}&&&&\\
&&Add& 9208&  \textbf{9208}&&&&\\ \Xhline{1pt}
\multirow{3}*{\makecell[c]{$16\times16$ \\$@2048$}}&\multirow{3}*{\makecell[c]{$5\times5$ \\$@64$}}&Rotation& 10752& \textbf{3132}&\multirow{3}*{41.3}&\multirow{3}*{42.0}&\multirow{3}*{\makecell[c]{\textbf{24.3}\\\textbf{(1.7$\times$)}}}&\multirow{3}*{\makecell[c]{\textbf{26.3}\\\textbf{(1.6$\times$)}}}\\
&&ScMult& 204800& \textbf{204800}&&&&\\
&&Add& 204796&  \textbf{204796}&&&&\\ \Xhline{1pt}
\end{tabular}}}
\end{table}

\subsection{Performance of Non-linear Layers}
\label{Cost Comparison of Non-linear Layers}
We compare the computation and communication overheads of SIMC and our SIMC 2.0 in the implementation of non-linear layers. We mainly show the cost required by the two schemes to securely execute different numbers of ReLU functions.

\subsubsection{Computation Overhead} \textbf{Figures}~\ref{FIGa} and \ref{FIGb} show the running time of SIMC and SIMC 2.0 for different numbers of ReLU functions under the LAN and WAN settings, respectively. We observe that SIMC 2.0 reduces the running time by about one-third compared to SIMC.  Our improvement mainly comes from the optimization mechanism designed for the computation of non-linear layers. As mentioned earlier, compared to SIMC, we only encapsulate the non-linear part of ReLU into the GC. The goal of this strategy is to substantially reduce the number of AND gates required in the GC (AND operations are usually expensive in the GC). For example, given GC's security parameter $\lambda=128$ bits and 44-bit prime field, SIMC needs  249 AND gates to calculate ReLU while we only need  161 AND gates. This saves 88 AND gate operations! Moreover, the alternative protocol we designed is also resource-friendly, requiring only the server and client to exchange two elements and perform simple mathematical operations locally.
\subsubsection{Communication Overhead}
\begin{figure}[htb]
  \centering
  \subfigure[]{\label{FIGa}\includegraphics[width=0.24\textwidth]{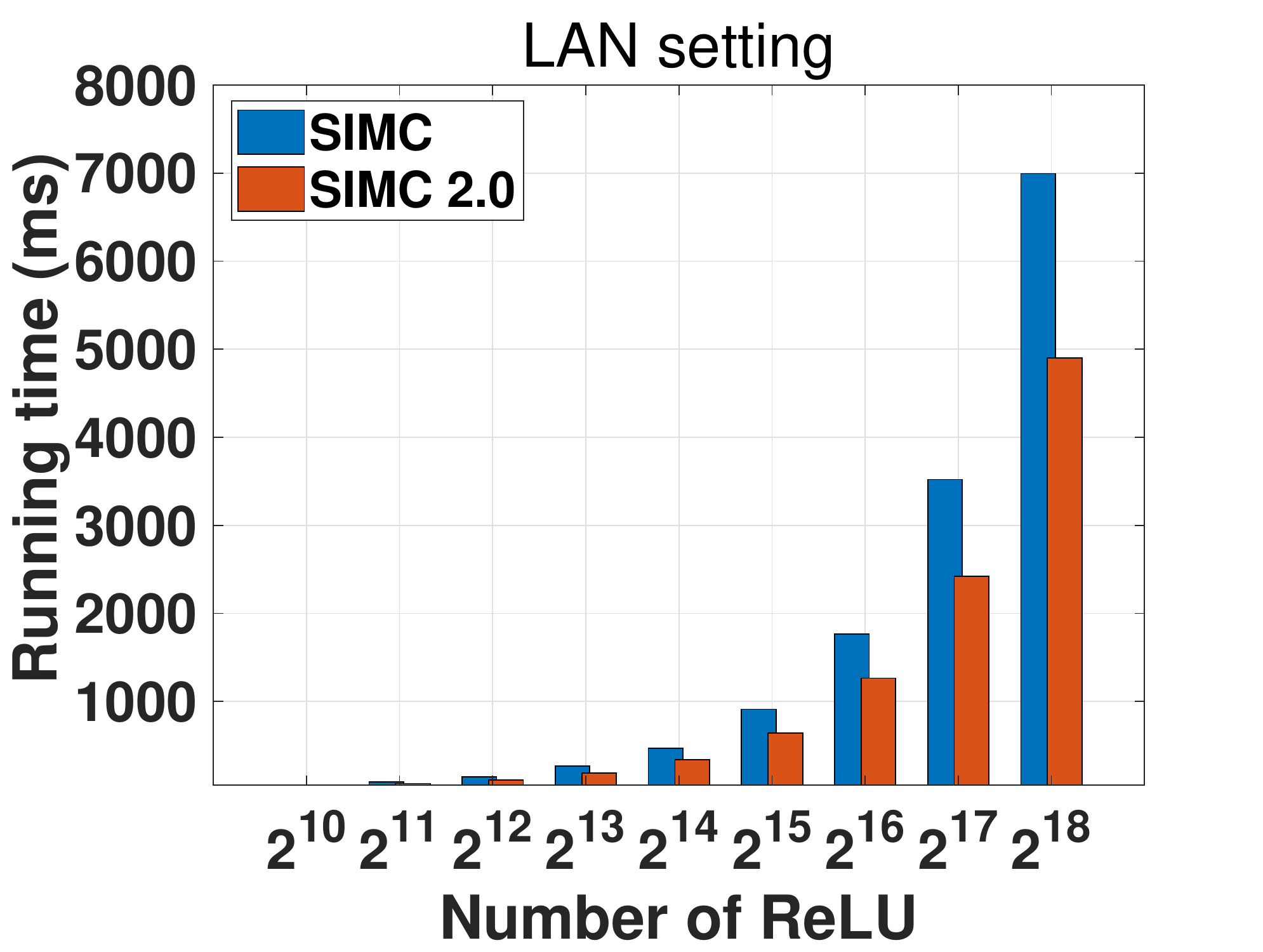}}
 \subfigure[]{\label{FIGb}\includegraphics[width=0.24\textwidth]{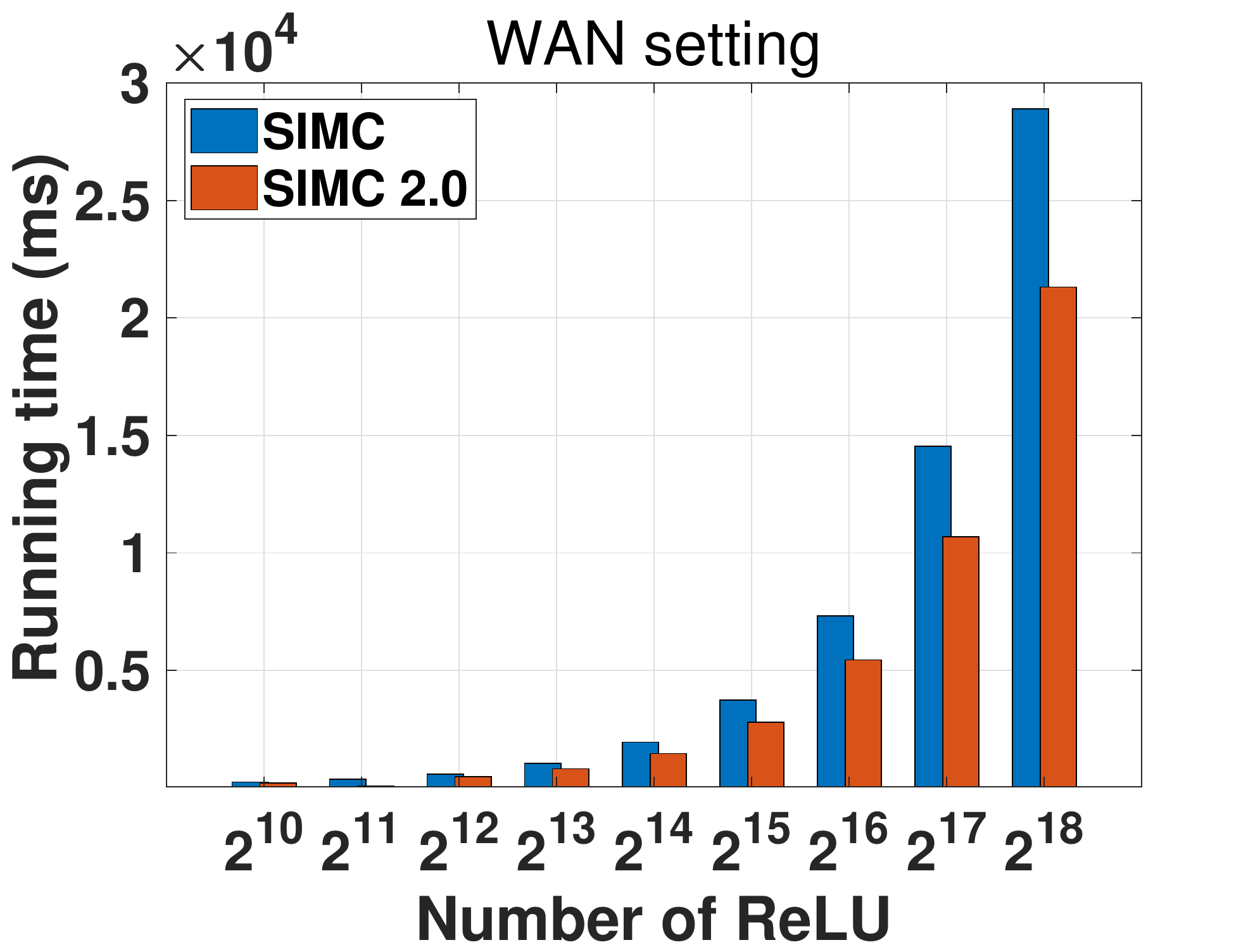}}
  \subfigure[]{\label{FIGc}\includegraphics[width=0.24\textwidth]{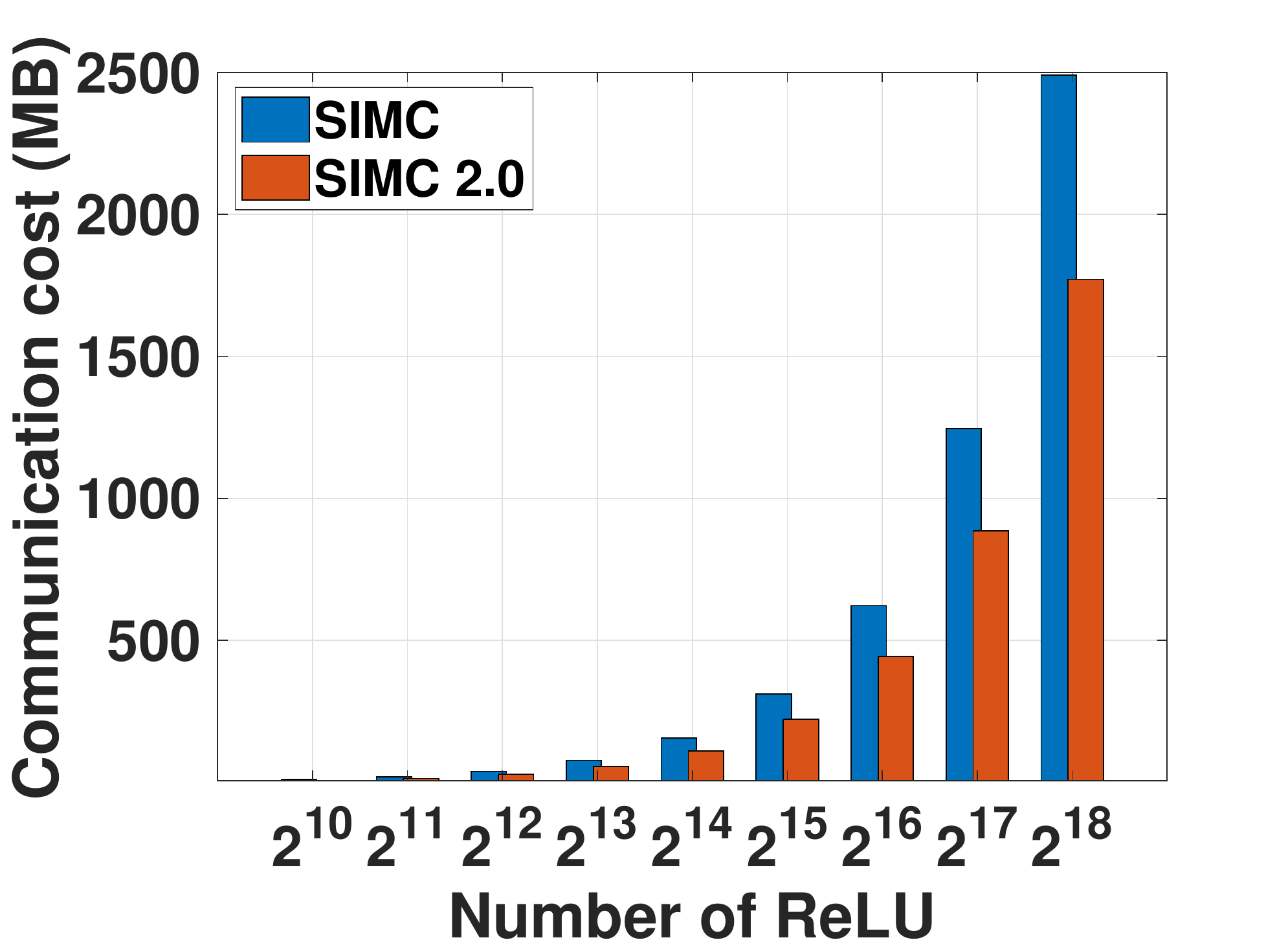}}
  \caption{Comparison of the nonlinear layer overhead. (a) Running time under the LAN setting.  (b) Running time under the WAN setting. (c) Communication overhead with different numbers of ReLU functions.}
  \label{Non-layer cost}
\end{figure}
\textbf{Figure}~\ref{FIGc} provides the communication overhead incurred by SIMC and our method in computing  different numbers of ReLU functions. We observe that SIMC 2.0 reduces the communication overhead of SIMC by approximately one-third. This also benefits from our optimized GC to perform non-linear operations. In more detail, given GC's security parameter $\lambda=128$ bits and 44-bit prime field, SIMC requires to communicate 11.92 KB data to perform one ReLU. In contrast, since our optimized GC saves about 88 AND gates, which needs 2.69 KB communication data for each ReLU, our SIMC 2.0 only requires 9.21 KB data transfer for each ReLU. Note that such communication advantages are non-trivial since mainstream DNN models usually contain tens of thousands of ReLU functions.

\subsection{Performance of End-to-end Secure Inference}
\label{Secure Inference Performance}
Finally we compare the computation and communication costs of SIMC and SIMC 2.0 for different complete DNN models. We employ 7 ML models in our experiments: (1) a multi-layer perceptron with the network structure of 784-128-128-10, which is often used for benchmarking private-preserving model deployment (\textit{e.g.}, SecureML \cite{mohassel2017secureml}, MiniONN \cite{liu2017oblivious}, GAZELLE). This model is trained with the MNIST dataset. (2) Some mainstream CNN models: AlexNet\cite{yu2016visualizing}, VGG-16 \cite{simonyan2014very}, ResNet-18 \cite{he2016deep}, ResNet-50 \cite{he2016deep}, ResNet-101 \cite{he2016deep}, and ResNet-152 \cite{he2016deep}. These models are trained over the CIFAR-10 dataset.
\begin{table}[htb]
\centering
\small
\caption{Cost of end-to-end model inference}
\label{Cost of Secure Inference}
 \setlength{\tabcolsep}{0.01mm}{\resizebox{\linewidth}{!}{
\begin{tabular}{c|c|c|c|c|c|c|c|c|c}
\Xhline{1pt}
\multirow{3}*{\textbf{Model}} &\multirow{3}*{\textbf{Metric}}&\multicolumn{2}{c|}{\textbf{\#operations}} &\multicolumn{4}{c}{\textbf{Running time (s)} }&\multicolumn{2}{|c}{\textbf{Comm. (GB)} }\\ \cline{3-10}
&&\multirow{2}*{\textbf{SIMC}}&\multirow{2}*{\textbf{Ours}}&\multicolumn{2}{c|}{\textbf{SIMC}}&\multicolumn{2}{c|}{\textbf{SIMC 2.0 (Speedup)}}&\multirow{2}*{\textbf{SIMC}}&\multirow{2}*{\textbf{Ours}}\\
&&&&LAN&WAN&LAN&WAN&&\\
\Xhline{1pt}
\multirow{3}*{MLP}&Rotation& 70& \textbf{55}&\multirow{3}*{0.14}&\multirow{3}*{0.21}&\multirow{3}*{\quad\makecell[c]{\textbf{0.11}\\\textbf{(1.3$\times$)}}\quad}&\multirow{3}*{\makecell[c]{\textbf{0.18}\\\textbf{(1.2$\times$)}}}&\multirow{3}*{<0.01}&\multirow{3}*{<0.01}\\
&ScMult& 56& \textbf{56}&&&&&&\\
&Add& 70&  \textbf{55}&&&&&&\\ \Xhline{1pt}
\multirow{3}*{AlexNet}&Rotation& 72549& \textbf{4394}&\multirow{3}*{81.0}&\multirow{3}*{85.8}&\multirow{3}*{\makecell[c]{\textbf{35.2}\\\textbf{(2.3$\times$)}}}&\multirow{3}*{\makecell[c]{\textbf{47,7}\\\textbf{(1.8$\times$)}}}&\multirow{3}*{0.42}&\multirow{3}*{0.29}\\
&ScMult& 931977& \textbf{931977}&&&&&&\\
&Add& 931643&  \textbf{931630}&&&&&&\\ \Xhline{1pt}
\multirow{3}*{VGG-16}&Rotation& 315030& \textbf{10689}&\multirow{3}*{72.4}&\multirow{3}*{97.6}&\multirow{3}*{\makecell[c]{\textbf{34.5}\\\textbf{(2.1$\times$)}}}&\multirow{3}*{\makecell[c]{\textbf{54.2}\\\textbf{(1.8$\times$)}}}&\multirow{3}*{5.81}&\multirow{3}*{4.13}\\
&ScMult& 3677802& \textbf{3677802}&&&&&&\\
&Add& 3676668&  \textbf{3676654}&&&&&&\\ \Xhline{1pt}
\multirow{3}*{ResNet-18}&Rotation& 234802& \textbf{9542}&\multirow{3}*{77.1}&\multirow{3}*{118.3}&\multirow{3}*{\makecell[c]{\textbf{35.0}\\\textbf{(2.2$\times$)}}}&\multirow{3}*{\makecell[c]{\textbf{65.6}\\\textbf{(1.8$\times$)}}}&\multirow{3}*{4.27}&\multirow{3}*{3.04}\\
&ScMult& 2737585& \textbf{2737585}&&&&&&\\
&Add&2736632&  \textbf{2736624}&&&&&&\\ \Xhline{1pt}
\multirow{3}*{ResNet-50}&Rotation& 2023606& \textbf{50666}&\multirow{3}*{434.8}&\multirow{3}*{683.9}&\multirow{3}*{\makecell[c]{\textbf{111.3}\\\textbf{(3.9$\times$)}}}&\multirow{3}*{\makecell[c]{\textbf{296.4}\\\textbf{(2.3$\times$)}}}&\multirow{3}*{29.50}&\multirow{3}*{21.00}\\
&ScMult& 5168181& \textbf{5168181}&&&&&&\\
&Add& 5162524&  \textbf{5162518}&&&&&&\\ \Xhline{1pt}
\multirow{3}*{ResNet-101}&Rotation& 3947190& \textbf{113770}&\multirow{3}*{802.3}&\multirow{3}*{1195.2}&\multirow{3}*{\makecell[c]{\textbf{186.5}\\\textbf{(4.3$\times$)}}}&\multirow{3}*{\makecell[c]{\textbf{478.1}\\\textbf{(2.5$\times$)}}}&\multirow{3}*{46.13}&\multirow{3}*{31.85}\\
&ScMult& 9903157& \textbf{9903157}&&&&&&\\
&Add& 9890972&  \textbf{9890964}&&&&&&\\ \Xhline{1pt}
\multirow{3}*{ResNet-152}&Rotation& 5533878& \textbf{169450}&\multirow{3}*{1209.2}&\multirow{3}*{1780.1}&\multirow{3}*{\makecell[c]{\textbf{274.8}\\\textbf{(4.4$\times$)}}}&\multirow{3}*{\makecell[c]{\textbf{684.6}\\\textbf{(2.6$\times$)}}}&\multirow{3}*{64.25}&\multirow{3}*{46.28}\\
&ScMult& 13802549& \textbf{13802549}&&&&&&\\
&Add& 13784604&  \textbf{13784596}&&&&&&\\ \Xhline{1pt}
\end{tabular}}}
\end{table}

\textbf{TABLE}~\ref{Cost of Secure Inference} provides the computation complexity of SIMC and our method for different models. It is clear that SIMC 2.0 reduces the number of rotation operations in SIMC by $16.5\times$, $29.4\times$, $24.6\times$, $39.9\times$, $34.6\times$ and $50.5\times$ for AlexNet, VGG-16, ResNet-18, ResNet-50, ResNet-10, and ResNet-152, respectively.  The fundamental reason for this improvement is that our designed optimization technique is for HE-based linear operations. We observe that SIMC 2.0 does not significantly reduce the number of rotation operations in the MLP model. It stems from the small ratio between the number of slots and the dimension of the output in the MLP, which limits the performance gain. \textbf{TABLE}~\ref{Cost of Secure Inference} also shows the running time and the corresponding speedup of our scheme. Since we substantially reduce the number of expensive rotation operations in the linear layer, and design an optimized GC for the nonlinear layer, our method obtains speedups of $2.3\times$, $2.1\times$, $2.2\times$, $3.9\times$, $4.3\times$ and $4.4\times$ for AlexNet, VGG-16, ResNet-18, ResNet-50, ResNet-101, and ResNet-152, respectively, under the LAN Setting. A similar performance boost is obtained under the WAN setting. Due to network latency in the WAN, the running time of both SIMC and our method is inevitably increased. This can reduce the performance advantage of SIMC 2.0, but our solution is still better than SIMC.

It has been demonstrated that for the HE-GC based secure inference model, the communication overhead caused by the non-linear layer dominates the communication in SIMC (refer to Seciton~5 in SIMC \cite{chandran2021simc}). In our experiments, the communication cost of linear layers in SIMC 2.0 remains consistent with SIMC, while the overhead of nonlinear layers is reduced by two-thirds. Therefore, we observe that the size of communication data required by SIMC 2.0 is still smaller than SIMC, remaining approximately between $2/3\sim3/4$ of the original size, depending on the model size. For example, for the ResNet-152 model, SIMC requires to transfer 64.25 GB data to perform an inference process while our SIMC 2.0 saves about 18GB in the communication overhead, which is quite impressive.

\section{Conclusion}
\label{sec:conclusion}
In this paper, we proposed SIMC 2.0, which significantly improves the performance of SIMC for secure ML inference in a threat model with a malicious client and honest-but-curious server. We designed a new coding method to minimize the number of costly rotation operations during homomorphic parallel matrix and vector computations in the linear layers. We also greatly reduced the size of the GC in SIMC to substantially speed up operations in the non-linear layers. In the future, we will focus on designing more efficient optimization strategies to further reduce the computation overhead of SIMC 2.0, to make secure ML inference more suitable for a wider range of practical applications.


\ifCLASSOPTIONcaptionsoff
\newpage \fi

\bibliographystyle{IEEEtran}
\bibliography{PPDR}

\begin{thebibliography}{10}
\providecommand{\url}[1]{#1}
\csname url@samestyle\endcsname
\providecommand{\newblock}{\relax}
\providecommand{\bibinfo}[2]{#2}
\providecommand{\BIBentrySTDinterwordspacing}{\spaceskip=0pt\relax}
\providecommand{\BIBentryALTinterwordstretchfactor}{4}
\providecommand{\BIBentryALTinterwordspacing}{\spaceskip=\fontdimen2\font plus
\BIBentryALTinterwordstretchfactor\fontdimen3\font minus
  \fontdimen4\font\relax}
\providecommand{\BIBforeignlanguage}[2]{{%
\expandafter\ifx\csname l@#1\endcsname\relax
\typeout{** WARNING: IEEEtran.bst: No hyphenation pattern has been}%
\typeout{** loaded for the language `#1'. Using the pattern for}%
\typeout{** the default language instead.}%
\else
\language=\csname l@#1\endcsname
\fi
#2}}
\providecommand{\BIBdecl}{\relax}
\BIBdecl

\bibitem{katz2018optimizing}
J.~Katz, S.~Ranellucci, and et~al, ``Optimizing authenticated garbling for
  faster secure two-party computation,'' in \emph{Annual International
  Cryptology Conference(CRYPTO)}.\hskip 1em plus 0.5em minus 0.4em\relax
  Springer, 2018, pp. 365--391.

\bibitem{wang2017faster}
X.~Wang, A.~J. Malozemoff, and J.~Katz, ``Faster secure two-party computation
  in the single-execution setting,'' in \emph{Annual International Conference
  on the Theory and Applications of Cryptographic Techniques(EUROCRYPT)}.\hskip
  1em plus 0.5em minus 0.4em\relax Springer, 2017, pp. 399--424.

\bibitem{rathee2020cryptflow2}
D.~Rathee, M.~Rathee, and et~al, ``Cryptflow2: Practical 2-party secure
  inference,'' in \emph{Proceedings of the ACM SIGSAC Conference on Computer
  and Communications Security (CCS)}, 2020, pp. 325--342.

\bibitem{patra2021aby2}
A.~Patra, T.~Schneider, and et~al, ``Aby2. 0: Improved mixed-protocol secure
  two-party computation,'' in \emph{USENIX Security Symposium}, 2021.

\bibitem{mohassel2017secureml}
P.~Mohassel and Y.~Zhang, ``Secureml: A system for scalable privacy-preserving
  machine learning,'' in \emph{2017 IEEE symposium on security and privacy
  (S\&P)}.\hskip 1em plus 0.5em minus 0.4em\relax IEEE, 2017, pp. 19--38.

\bibitem{lavigne2018topology}
R.~LaVigne, C.-D. Liu-Zhang, and et~al, ``Topology-hiding computation beyond
  semi-honest adversaries,'' in \emph{Theory of Cryptography Conference
  (TCC)}.\hskip 1em plus 0.5em minus 0.4em\relax Springer, 2018, pp. 3--35.

\bibitem{hussain2021coinn}
S.~U. Hussain, M.~Javaheripi, and et~al, ``Coinn: Crypto/ml codesign for
  oblivious inference via neural networks,'' in \emph{Proceedings of the ACM
  SIGSAC Conference on Computer and Communications Security (CCS)}, 2021, pp.
  3266--3281.

\bibitem{lou2021hemet}
Q.~Lou and L.~Jiang, ``Hemet: A homomorphic-encryption-friendly
  privacy-preserving mobile neural network architecture,'' in
  \emph{International Conference on Machine Learning (ICML)}.\hskip 1em plus
  0.5em minus 0.4em\relax PMLR, 2021, pp. 7102--7110.

\bibitem{lehmkuhl2021muse}
R.~Lehmkuhl, P.~Mishra, and et~al, ``Muse: Secure inference resilient to
  malicious clients,'' in \emph{USENIX Security Symposium}, 2021.

\bibitem{chandran2021simc}
N.~Chandran, D.~Gupta, and et~al, ``Simc: Ml inference secure against malicious
  clients at semi-honest cost,'' in \emph{USENIX Security Symposium}, 2022.

\bibitem{patrablaze}
A.~Patra and A.~Suresh, ``Blaze: Blazing fast privacy-preserving machine
  learning,'' in \emph{Proceedings of the Network and Distributed System
  Security (NDSS)}, 2020.

\bibitem{damgaard2019new}
I.~Damg{\aa}rd, D.~Escudero, and et~al, ``New primitives for actively-secure
  mpc over rings with applications to private machine learning,'' in \emph{IEEE
  Symposium on Security and Privacy (S\&P)}.\hskip 1em plus 0.5em minus
  0.4em\relax IEEE, 2019, pp. 1102--1120.

\bibitem{escudero2020improved}
D.~Escudero, S.~Ghosh, and et~al, ``Improved primitives for mpc over mixed
  arithmetic-binary circuits,'' in \emph{Annual International Cryptology
  Conference (CRYPTO)}.\hskip 1em plus 0.5em minus 0.4em\relax Springer, 2020,
  pp. 823--852.

\bibitem{hazay2019constant}
C.~Hazay and A.~Yanai, ``Constant-round maliciously secure two-party
  computation in the ram model,'' \emph{Journal of Cryptology}, vol.~32, no.~4,
  pp. 1144--1199, 2019.

\bibitem{mishra2020delphi}
P.~Mishra, R.~Lehmkuhl, and et~al, ``Delphi: A cryptographic inference service
  for neural networks,'' in \emph{USENIX Security Symposium}, 2020, pp.
  2505--2522.

\bibitem{zhang2021gala}
Q.~Zhang, C.~Xin, and H.~Wu, ``Gala: Greedy computation for linear algebra in
  privacy-preserved neural networks,'' in \emph{Proceedings of the Network and
  Distributed System Security (NDSS)}, 2021.

\bibitem{sav2020poseidon}
S.~Sav, A.~Pyrgelis, and et~al, ``Poseidon: Privacy-preserving federated neural
  network learning,'' in \emph{Proceedings of the Network and Distributed
  System Security (NDSS)}, 2021.

\bibitem{huang2022cheetah}
Z.~Huang, W.-j. Lu, C.~Hong, and J.~Ding, ``Cheetah: Lean and fast secure
  two-party deep neural network inference,'' \emph{Cryptology ePrint Archive},
  2022.

\bibitem{jiang2018secure}
X.~Jiang, M.~Kim, and et~al, ``Secure outsourced matrix computation and
  application to neural networks,'' in \emph{Proceedings of the ACM SIGSAC
  Conference on Computer and Communications Security (CCS)}, 2018, pp.
  1209--1222.

\bibitem{juvekar2018gazelle}
C.~Juvekar, V.~Vaikuntanathan, and et~al, ``$\{$GAZELLE$\}$: A low latency
  framework for secure neural network inference,'' in \emph{USENIX Security
  Symposium}, 2018, pp. 1651--1669.

\bibitem{chen2020maliciously}
H.~Chen, M.~Kim, and et~al, ``Maliciously secure matrix multiplication with
  applications to private deep learning,'' in \emph{International Conference on
  the Theory and Application of Cryptology and Information
  Security(ASIACRYPT)}.\hskip 1em plus 0.5em minus 0.4em\relax Springer, 2020,
  pp. 31--59.

\bibitem{chandran2019ezpc}
N.~Chandran, D.~Gupta, and et~al, ``Ezpc: programmable and efficient secure
  two-party computation for machine learning,'' in \emph{IEEE European
  Symposium on Security and Privacy (EuroS\&P)}.\hskip 1em plus 0.5em minus
  0.4em\relax IEEE, 2019, pp. 496--511.

\bibitem{demmler2015aby}
D.~Demmler, T.~Schneider, and M.~Zohner, ``Aby-a framework for efficient
  mixed-protocol secure two-party computation.'' in \emph{Proceedings of the
  Network and Distributed System Security (NDSS)}, 2015.

\bibitem{wang2017authenticated}
X.~Wang, S.~Ranellucci, and J.~Katz, ``Authenticated garbling and efficient
  maliciously secure two-party computation,'' in \emph{Proceedings of the ACM
  SIGSAC Conference on Computer and Communications Security (CCS)}, 2017, pp.
  21--37.

\bibitem{halevi2014algorithms}
S.~Halevi and V.~Shoup, ``Algorithms in helib,'' in \emph{Annual Cryptology
  Conference (CRYPTO)}.\hskip 1em plus 0.5em minus 0.4em\relax Springer, 2014,
  pp. 554--571.

\bibitem{keller2015actively}
M.~Keller, E.~Orsini, and P.~Scholl, ``Actively secure ot extension with
  optimal overhead,'' in \emph{Annual Cryptology Conference (CRYPTO)}.\hskip
  1em plus 0.5em minus 0.4em\relax Springer, 2015, pp. 724--741.

\bibitem{ghodsi2021circa}
Z.~Ghodsi, N.~K. Jha, B.~Reagen, and S.~Garg, ``Circa: Stochastic relus for
  private deep learning,'' \emph{Advances in Neural Information Processing
  Systems(NeurIPS)}, vol.~34, 2021.

\bibitem{zahur2015two}
S.~Zahur, M.~Rosulek, and D.~Evans, ``Two halves make a whole,'' in
  \emph{Annual International Conference on the Theory and Applications of
  Cryptographic Techniques (EUROCRYPT)}.\hskip 1em plus 0.5em minus 0.4em\relax
  Springer, 2015, pp. 220--250.

\bibitem{sealcrypto}
``{M}icrosoft {SEAL} (release 3.3),'' \url{https://github.com/Microsoft/SEAL},
  Jun. 2019, microsoft Research, Redmond, WA.

\bibitem{wang2016emp}
X.~Wang, A.~J. Malozemoff, and J.~Katz, ``Emp-toolkit: Efficient multiparty
  computation toolkit,'' \url{https://github.com/emp-toolkit}, 2016.

\bibitem{liu2017oblivious}
J.~Liu, M.~Juuti, Y.~Lu, and N.~Asokan, ``Oblivious neural network predictions
  via minionn transformations,'' in \emph{Proceedings of the 2017 ACM SIGSAC
  conference on computer and communications security (CCS)}, 2017, pp.
  619--631.

\bibitem{yu2016visualizing}
W.~Yu, K.~Yang, and et~al, ``Visualizing and comparing alexnet and vgg using
  deconvolutional layers,'' in \emph{Proceedings of the International
  Conference on Machine Learning (ICML)}, 2016.

\bibitem{simonyan2014very}
K.~Simonyan and A.~Zisserman, ``Very deep convolutional networks for
  large-scale image recognition,'' in \emph{Proceedings of the International
  Conference on Learning Representations (ICLR)}, 2015.

\bibitem{he2016deep}
K.~He, X.~Zhang, S.~Ren, and J.~Sun, ``Deep residual learning for image
  recognition,'' in \emph{Proceedings of the IEEE conference on computer vision
  and pattern recognition (CVPR)}, 2016, pp. 770--778.

\bibitem{lindell2017simulate}
Y.~Lindell, ``How to simulate it--a tutorial on the simulation proof
  technique,'' \emph{Tutorials on the Foundations of Cryptography}, pp.
  277--346, 2017.

\bibitem{keller2018overdrive}
M.~Keller, V.~Pastro, and D.~Rotaru, ``Overdrive: Making spdz great again,'' in
  \emph{Annual International Conference on the Theory and Applications of
  Cryptographic Techniques(EUROCRYPT)}.\hskip 1em plus 0.5em minus 0.4em\relax
  Springer, 2018, pp. 158--189.

\end{thebibliography}

\clearpage

\appendices
\balance
\section*{Appendix}
\label{sec:APPENDIX}
\setcounter{section}{0}

\section{Threat Model}
\label{A:threat model}
We present the formal security using the simulation paradigm \cite{lindell2017simulate}. Security is modeled through two interactive protocols: a real interactive protocol where $S_0$ and $S_1$ execute the protocol in the presence of an adversary $\mathbf{A}$ and an environment $Z$, and an ideal interactive protocol where the parties send input to a trusted third party who performs the computation faithfully. Security requires that for  any adversary $\mathbf{A}$, there is a simulator $\mathbf{S}$ in the ideal interaction, which enables no environment $\mathbf{Z}$ to distinguish between the ideal and the real interactions. Specifically, Let $f=(f_0, f_1)$ be the functionality of the two parties so that $S_0$ learns $f_0(x, y)$ and $S_1$ learns $f_1(x, y)$ with $x$ and $y$ as inputs, respectively. Hence, we say that a protocol $\pi$ securely implements $f$ on the $client-malicious$ adversary model  if the following properties are held.
\begin{itemize}

\item  \textbf{Correctness}: If $S_0$ and $S_1$ are both honest, then $S_0$ gets $f_0(x, y)$ and $S_1$ gets $f_1(x,y)$ from the execution of $\pi$ on the inputs  $x$ and $y$, respectively.
\item \textbf{Semi-honest Server Security}: For the semi-honest adversary $S_0$, there exists a simulator $\mathbf{S}$, so that for any input $(x, y)$,  we have
\begin{small}
\begin{equation*}
\begin{split}
View_{\mathbf{A}}^{\pi}(x, y)\approx \mathbf{S}(x, f_0(x, y))
\end{split}
\end{equation*}
\end{small}
where $View_{\mathbf{A}}^{\pi}(x, y)$ denotes $\mathbf{A}$'s view during the execution of $\pi$ given the $S_0$'s input $x$ and $S_1$'s input $y$.  $\mathbf{S}(x, f_0(x, y))$  represents the view simulated by $\mathbf{S}$ when it is given access to  $x$ and  $f_0(x, y)$ of $S_0$. $\approx$ indicates computational indistinguishability
of   two distributions $View_{\mathbf{A}}^{\pi}(x, y)$ and $\mathbf{S}(x, f_0(x, y))$.
 \item \textbf{Malicious Client Security}: For the malicious adversary $S_1$, there exists a simulator $\mathbf{S}$, so that for any input $x$ from $S_0$,  we have
 \begin{small}
\begin{equation*}
\begin{split}
Out_{S_0}, View_{\mathbf{A}}^{\pi}(x,\cdot)\approx \hat{Out}, \mathbf{S}^{f(x, \cdot)}
\end{split}
\end{equation*}
\end{small}
\end{itemize}
where $View_{\mathbf{A}}^{\pi}(x, \cdot)$ denotes $\mathbf{A}$'s view during the execution of $\pi$ given the $S_0$'s input $x$. $Out_{S_0}$ indicates the output of $S_0$ in the same protocol execution. Similarly,   $\hat{Out}$  and $\mathbf{S}^{f(x, \cdot)}$  represents the output of $S_0$  and the view simulated by $\mathbf{S}$ in an ideal interaction.
\section{Generation of Authenticated Beaver's Multiplicative Triple}
\label{A:Authenticated Beaver's multiplicative triples}

\renewcommand\tablename{Figure}
\renewcommand \thetable{\arabic{table}}
\setcounter{table}{6}
\begin{table}[htb]
\centering
\begin{tabular}{|p{7.5cm}|}
\hline \\
 \textbf{Input:} $S_1$ holds $(\langle A\rangle_1, \langle B\rangle_1)$ uniformly chosen from $\mathbb{F}_p$.   $S_0$ holds $(\langle A\rangle_0, \langle B\rangle_0)$ and a MAC key $\alpha$ uniformly chosen from $\mathbb{F}_p$\\
 \textbf{Output:} $S_b$ obtains $\{(\langle A\rangle_b, $ $\langle \alpha A\rangle_b), $ $(\langle B\rangle_b, \langle \alpha B\rangle_b), $ $(\langle C\rangle_b, $ $\langle \alpha C\rangle_b)\}$ for $b\in \{0, 1\}$.\\
 \textbf{Procedure}:
 \begin{itemize}
    \item[1.]  $S_0$ and $S_1$ participate in a two-party secure computing protocol against the semi-honest server $S_0$ and the malicious client $S_1$, so that $S_1$ obtains an FHE's public and secret key $(pk, sk)$ and $S_0$ obtains $pk$. This process is performed only once.
    \item[2.] $S_1$ send the encryptions $c_1\leftarrow \mathtt{Enc}(pk, A_1)$  and $c_2\leftarrow \mathtt{Enc}(pk, B_1)$ to $S_0$ along with a zero-knowledge (ZK) proof of plaintext knowledge of the two ciphertexts. A ZK proof of knowledge for  ciphertexts is used to state that $c_{1}$ and $c_{2}$ are  valid ciphertexts generated from the given FHE cryptosystem. Readers can refer to \cite{keller2018overdrive,chen2020maliciously} for more details.
    \end{itemize}
 \begin{itemize}
    \item[3.] $S_0$ samples $(\langle \alpha A\rangle_0, \langle \alpha B\rangle_0, \langle \alpha C\rangle_0,  \langle C\rangle_0) $ from $\mathbb{F}_{p}^{3}$, and  then computes $c_3=\mathtt{Enc}_{pk}( \alpha(\langle A\rangle_0+\langle A\rangle_1)-\langle \alpha A\rangle_0)$, $c_4=\mathtt{Enc}_{pk}( \alpha(\langle B\rangle_0+\langle B\rangle_1)-\langle \alpha B\rangle_0)$, $c_5=\mathtt{Enc}_{pk}(\alpha(A\cdot B)-\langle \alpha C\rangle_0)$, and $c_6=\mathtt{Enc}_{pk}((A\cdot B)-\langle C\rangle_0)$.
    \item[4.] $S_1$ decrypts $c_3$, $c_4$,  $c_5$ and   $c_6$ to obtain $(\langle \alpha A\rangle_1, \langle \alpha B\rangle_1, \langle \alpha C\rangle_1, \langle C\rangle_1)$.
    \item[5.] $S_b$ outputs $\{(\langle A\rangle_b, $ $\langle \alpha A\rangle_b), $ $(\langle B\rangle_b, \langle \alpha B\rangle_b), $ $(\langle C\rangle_b, $ $\langle \alpha C\rangle_b)\}$ for $b\in \{0, 1\}$.
    \end{itemize}\\
\hline
\end{tabular}
\caption{Algorithm of generating authenticated Beaver's multiplicative triple}
\label{Algorithm of generating authenticated Beaver's multiplicative triple}
\end{table}

\section{Protocols of InitLin and Lin}
\label{A:Protocols InitLin and Lin}
\begin{table}[htb]
\centering
\begin{tabular}{|p{7.5cm}|}
\hline \\
 \textbf{Input:} $S_0$ holds $\mathbf{N}\in \mathbb{F}_{p}^{n_o\times n_i}$  and  a MAC key $\alpha$ uniformly chosen from $\mathbb{F}_p$.  $S_1$ holds  $\mathbf{t}\in \mathbb{F}_{p}^{n_i}$.\\
 \textbf{Output:} $S_b$ obtains $\{(\langle \mathbf{Nt}\rangle_b, \langle \alpha \mathbf{Nt}\rangle_b)\}$ for $b\in\{0, 1\}$.\\
 \textbf{Procedure}:
 \begin{itemize}
    \item[1.]  $S_0$ and $S_1$ participate in a two-party secure computing protocol against the semi-honest server $S_0$ and the malicious client $S_1$, so that $S_1$ obtains the FHE's public and secret key $(pk, sk)$ and $S_0$ obtains $pk$. This process is performed only once.
    \item[2.] $S_1$ sends the encryption $\mathbf{c_1}\leftarrow \mathtt{Enc}(pk, \mathbf{t})$  to $S_0$ along with a zero-knowledge (ZK) proof of plaintext knowledge of this ciphertext.
    \end{itemize}
 \begin{itemize}
    \item[3.] $S_0$ samples $(\langle \mathbf{Nt}\rangle_0, \langle \alpha \mathbf{Nt}\rangle_0)$ from $\mathbb{F}_{p}^{n_o}$, and  then sends $\mathbf{c_2}=\mathtt{Enc}_{pk}(\mathbf{Nt}-\langle \mathbf{Nt}\rangle_0)$ and $\mathbf{c_3}=\mathtt{Enc}_{pk}(\alpha\mathbf{Nt}-\langle \alpha\mathbf{Nt}\rangle_0)$ to $S_1$.
    \item[4.] $S_1$ decrypts $\mathbf{c_2}$ and $\mathbf{c_3}$ and set $\langle \mathbf{Nt}\rangle_1=\mathtt{Dec}_{sk}(\mathbf{c_2})$, $\langle \alpha\mathbf{Nt}\rangle_1=\mathtt{Dec}_{sk}(\mathbf{c_3})$.
    \item[5.] $S_b$ outputs $\{(\langle \mathbf{Nt}\rangle_b, \langle \alpha \mathbf{Nt}\rangle_b)\}$ for $b\in\{0, 1\}$.
    \end{itemize}\\
\hline
\end{tabular}
\caption{Protocol of \textbf{InitLin}}
\label{Protocols InitLin}
\end{table}

\begin{table}[htb]
\centering
\begin{tabular}{|p{7.5cm}|}
\hline \\
 \textbf{Input:} $S_0$ holds $\mathbf{N}\in \mathbb{F}_{p}^{n_o\times n_i}$, a MAC key $\alpha$ uniformly chosen from $\mathbb{F}_p$, and $\langle \mathbf{t}\rangle_0$, $\langle \mathbf{d}\rangle_0 \in \mathbb{F}_{p}^{n_i}$.  $S_1$ holds  $\langle \mathbf{t}\rangle_1, \langle \mathbf{d}\rangle_1 \in \mathbb{F}_{p}^{n_i}$, where $\mathbf{d}=\alpha \mathbf{t}$. \\
 \textbf{Output:} $S_b$ obtains $\{\langle \mathbf{Nt}\rangle_b, \langle \mathbf{Nd}\rangle_b, \langle \alpha^{3}\mathbf{t}-\alpha^{2}\mathbf{d}\rangle_b\}$ for $b\in\{0, 1\}$.\\
 \textbf{Procedure}:
 \begin{itemize}
    \item[1.] $S_1$ sends the encryptions $\mathbf{e_1}\leftarrow \mathtt{Enc}(pk, \langle \mathbf{t}\rangle_1)$  and $\mathbf{e_2}\leftarrow \mathtt{Enc}(pk, \langle \mathbf{d}\rangle_1)$ to $S_0$ along with a zero-knowledge (ZK) proof of plaintext knowledge of the two ciphertexts.
    \end{itemize}
 \begin{itemize}
    \item[2.] $S_0$ samples $(\langle \mathbf{Nt}\rangle_0, \langle \mathbf{Nd}\rangle_0)$ from $\mathbb{F}_{p}^{n_o}$ and $\langle \alpha^{3}\mathbf{u}-\alpha^{2}\mathbf{d}\rangle_0$ from $\mathbb{F}_{p}^{n_i}$.

    \item[3.] $S_0$ computes $\mathbf{e_3}=\mathtt{Enc}_{pk}(\mathbf{Nt}-\langle \mathbf{Nt}\rangle_0)$, $\mathbf{e_4}=\mathtt{Enc}_{pk}(\mathbf{Nd}-\langle \mathbf{Nd}\rangle_0)$, and $ \mathbf{e_5}=\alpha^{3}\mathbf{t}-\alpha^{2}\mathbf{d}-\langle \alpha^{3}\mathbf{t}-\alpha^{2}\mathbf{d}\rangle_0$. $S_0$ sends $\mathbf{e_3}$, $\mathbf{e_4}$ and $\mathbf{e_5}$ to $S_1$.
    \item[4.] $S_1$ sets  $\langle \mathbf{Nt}\rangle_1=\mathtt{Dec}_{sk}(\mathbf{e_3})$, $\langle \mathbf{Nd}\rangle_1=\mathtt{Dec}_{sk}(\mathbf{e_4})$, and $\langle \alpha^{3}\mathbf{t}-\alpha^{2}\mathbf{d}\rangle_1=\mathtt{Dec}_{sk}(\mathbf{e_5})$.
    \item[5.] $S_b$ outputs $\{\langle \mathbf{Nt}\rangle_b, \langle \mathbf{Nd}\rangle_b, \langle \alpha^{3}\mathbf{t}-\alpha^{2}\mathbf{d}\rangle_b\}$ for $b\in\{0, 1\}$.
    \end{itemize}\\
\hline
\end{tabular}
\caption{Protocol of \textbf{Lin}}
\label{Protocol Lin}
\end{table}
\section{Optimization for Convolution Operations}
\label{Optimization for convolution operations}

We show how to compute convolutional layers in a parallelized manner, where we mainly incorporate the core idea of work GALA \cite{zhang2021gala}. We reuse the symbols $\mathbf{N}$ and $\mathbf{[t]_c}$ used in computing the matrix-vector multiplication, which represent the inputs held by the server and client, respectively. To be precise, we assume that the server has $c_o$ plaintext kernels of size $k_w\times k_h\times c_i$, and the ciphertext input sent by the client to the server is $c_i$ kernels of size $u_w\times u_h$. The server is required to perform  homomorphic convolution operations between the ciphertext input and its own plaintext kernel to obtain the ciphertext output. We first describe the basic convolution operation under single-input single-output (SISO), and then we extend this to the general case of multiple-input multiple-output (MIMO) (the method designed in GAZELLE \cite{juvekar2018gazelle}). Finally, we show the details of our method which exhibits better computational performance.

\subsection{Basic Convolution for the Single Input Single Output (SISO)}
\renewcommand\tablename{TABLE}
\renewcommand \thetable{\Roman{table}}
\setcounter{table}{4}
\setcounter{figure}{9}
\begin{figure}[htb]
\centering
\includegraphics[width=0.5\textwidth]{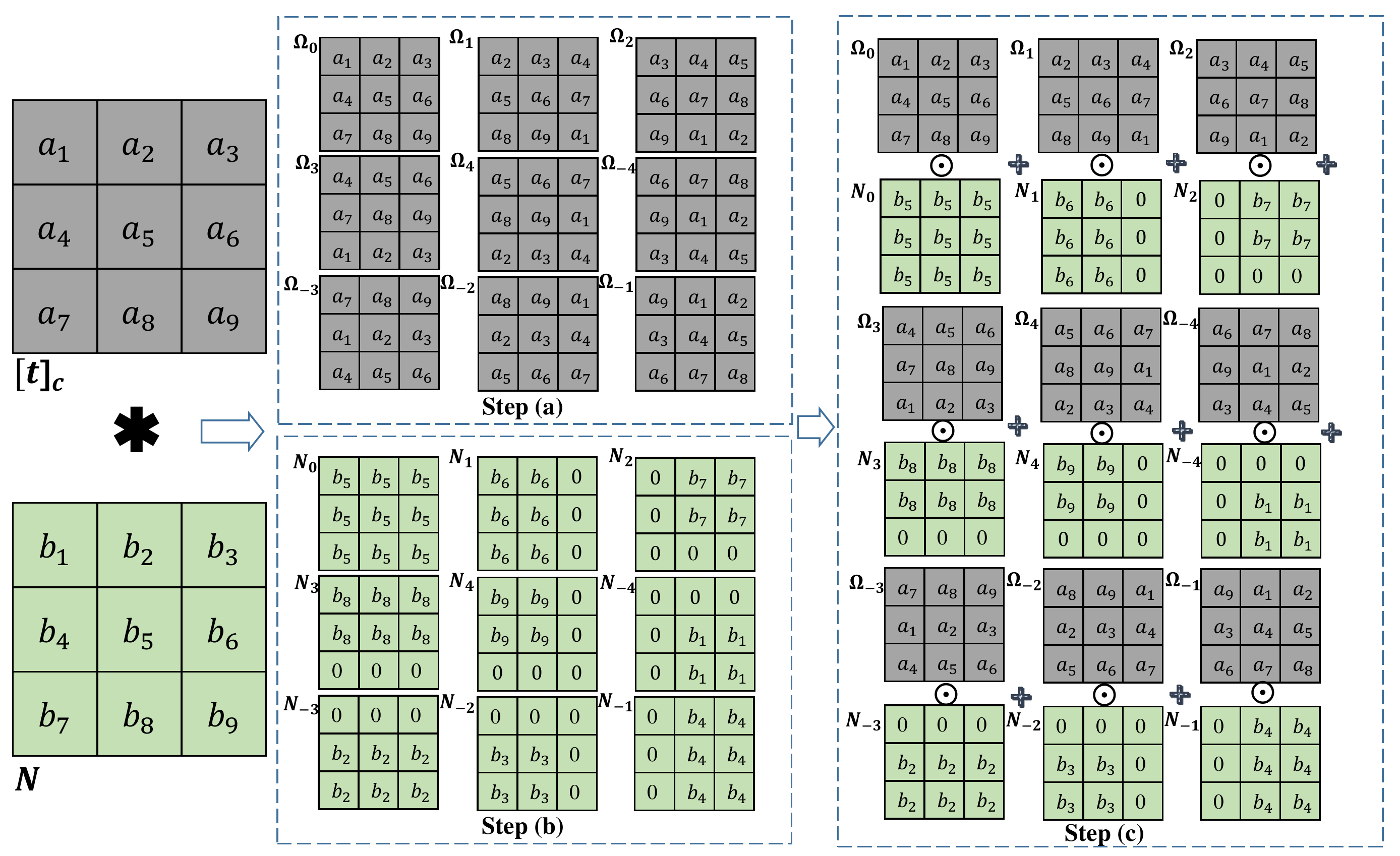}
\caption{SISO convolution}
\label{Fig:SISO convolution}
\end{figure}

SISO is a special case of MIMO where $c_i=c_o=1$.  Therefore, the client's encrypted input is a data (\textit{i.e.}, a 2D image) of size $u_w\times u_h$ with one channel and the server holds a 2D filter of size $k_w\times k_h$ with one kernel. \textbf{Figure}~\ref{Fig:SISO convolution} gives an example where $\mathbf{[t]_c}$ is the input from the client and $\mathbf{N}$ represents the plaintext kernel held by the server. The convolution computation process can be visualized as first placing $\mathbf{N}$ at different positions of the $\mathbf{[t]_c}$. Then, for each location, the sum of the element-wise products of the kernel and the corresponding data value within the kernel window is computed. As shown in \textbf{Figure}~\ref{Fig:SISO convolution}, the first value of the convolution between $\mathbf{[t]_c}$ and  $\mathbf{N}$ is ($a_1b_5+a_2b_6+a_4b_8+a_5b_9$). It is obtained by first placing the center of $\mathbf{N}$, \textit{i.e.} $b_5$ at the position of $a_1$, and then computing the element-wise product between $\mathbf{N}$ and the part of $\mathbf{[t]_c}$ that is within $\mathbf{N}$'s kernel window (\textit{i.e.}, $a_1$, $a_2$, $a_4$ and $a_5$). Similarly, the remaining convolution results can be computed sequentially by placing $b_5$  in $a_2$ to $a_9$.

We take \textbf{Figure}~\ref{Fig:SISO convolution} as an example to illustrate the basic SISO procedure in detail, where $b_5$ is first placed in the position of $a_5$. We observe that the kernel size is $k_w\times k_h=9$. The final convolution result is obtained by summing the element-wise products between the 9 values in $\mathbf{N}$ and the corresponding 9 values around $a_5$. This can be done by rotating $\mathbf{[t]_c}$ in a raster scan manner. Specifically, we first convert $\mathbf{[t]_c}$ to the equivalent vector format and perform ($k_w\times k_h-1$) round rotations on it, with the first half forward and the remaining half in the backward direction. We use the notation $\Omega_j$ to indicate that $\mathbf{[t]_c}$ is rotated by $j$ positions, where a positive $j$ represents the forward direction and a negative $j$ represents the backward direction (as shown in Step(a) in Figure \textbf{Figure}~\ref{Fig:SISO convolution}).

When given all $\Omega_j$, the coefficients of the plaintext kernel $\mathbf{N}$ also need to be re-encoded accordingly (See $\mathbf{N}_j$ in step(b) in \textbf{Figure}~\ref{Fig:SISO convolution}) based on the procedure of convolution operation. Then, each $\Omega_j$ performs an element-wise product with the corresponding $\mathbf{N}_j$  and we get the final encrypted results by summing up all the ciphertexts (shown in step(c) in \textbf{Figure}~\ref{Fig:SISO convolution}). As for the computation complexity, SISO requires a total of $(k_wk_h-1)$ rotation operations (excluding the trivial rotation by zero), all of which are operated on $\mathbf{[t]_c}$, and $(k_wk_h)$ times ScMult operations and $(k_wk_h-1)$ times Add operations.
\subsection{MIMO Convolution (GAZELLE)}
\begin{figure}[htb]
\centering
\includegraphics[width=0.5\textwidth]{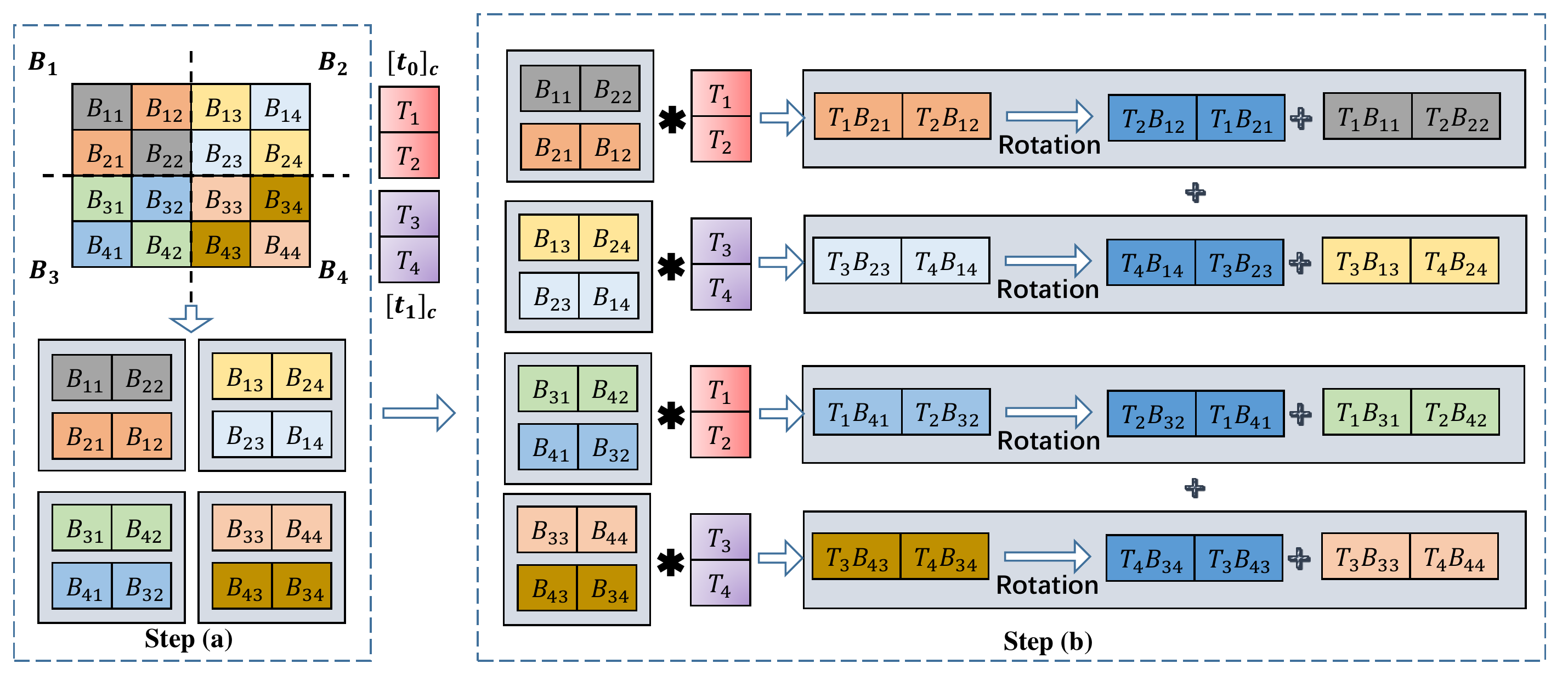}
\caption{MIMO convolution}
\label{Fig:MIMO convolution}
\end{figure}
We now consider the general case, \textit{i.e.}, MIMO, where $c_i$ or $c_o$ is not one. A simple way is to use the above convolution operation for SIMO multiple times. Specifically, the client first encrypts $c_i$ input channels into $c_i$ ciphertexts, $\{\mathbf{[t_i]_c}\}$ and sends them to the server. Since each of $c_o$ channels contains $c_i$ filters, each $\mathbf{[t_i]_c}$ is convolved with one of the $c_i$ filters by way of SISO. As a result, the final ciphertext result will be obtained by summing all of the $c_i$ SISO convolutions. Obviously, this naive approach requires $c_i(k_wk_h-1)$ rotation operations for $c_i$ input channels, $c_i(k_wk_h)$ ScMult operations, and $c_o(c_ik_wk_h-1)$ addition operations. This method outputs $c_o$ ciphertexts.

GAZELLE presents a new MIMO convolution method that substantially reduces the computation complexity compared to the naive method. Its key insight is the efficient utilization of the ciphertext slot $n$, due to the observation that $n$ is usually much larger than the channel size $u_wu_h$. Specifically, to improve utilization and parallel processing capabilities, GAZELLE proposes to pack the input data of $c_n$ channels into one ciphertext, that is, the input is now $\frac{c_i}{c_n}$ ciphertexts instead of $c_i$ ciphertexts (see \textbf{Figure}~\ref{Fig:MIMO convolution} where the four input channels are organized into two ciphertexts, each of which contains 2 channels). As a result, the $c_o$ kernels held by the server can be viewed as a $c_o\times c_i$ size kernel block, where each row of the block contains $c_i$ 2D filters for each kernel. Since then, the convolution operation of MIMO can be analogized to matrix-vector multiplication, where element-wise multiplication is replaced by convolution. Since each ciphertext contains $c_n$ channels, the entire kernel block can be further divided into $\frac{c_0c_i}{c_n^2}$ blocks (as shown in Step(a) in \textbf{Figure}~\ref{Fig:MIMO convolution}, where the kernel block is split into $\mathbf{B_1}$ to $\mathbf{B_4}$).

Each divided block will be further encoded as vectors in a diagonal manner (see Step(a) in \textbf{Figure}~\ref{Fig:MIMO convolution}, where we rearrange the positions of elements in each divided block $\mathbf{B_i}$). Afterwards, we can execute convolution operations between each input ciphertext and vectors in each divided block with the SISO manner.  The $c_n$ convolved vectors will be rotated diagonally to coincide with the elements in the corresponding kernel block, thereby obtaining the convolution of each divided block with the corresponding vector. Finally, we obtain the ciphertext result by summing up all results ciphertexts(see Step(b) in \textbf{Figure}~\ref{Fig:MIMO convolution}).

We now analyze the computation complexity introduced by GAZELLE. We observe that for each of the $\frac{c_oc_i}{c_n^2}$ blocks, there are a total of $(c_n-1)$ rotation operations.  This requires $c_nk_wk_h$ ScMult operations and $(c_nk_wk_h-1)$ Add operations. Additionally, there are a total of $\frac{c_i}{c_n}$ block convolutions associated with the same kernel order. Since for each of the $\frac{c_i}{c_n}$ blocks, we use the SISO method to calculate the convolution between the corresponding kernel blocks, which requires a total of $\frac{c_i(k_wk_h-1)}{c_n}$ rotation operations. In general, MIMO convolution requires a total of $\frac{c_oc_i}{c_n^2}(c_n-1)+\frac{c_i(k_wk_h-1)}{c_n}$ rotations, $k_wk_h\frac{c_ic_o}{c_n}$ ScMult and $\frac{c_o}{c_n}(c_ik_wk_h-1)$ Add operations. This method outputs $\frac{c_o}{c_n}$ ciphertexts.

\subsection{Our MIMO Convolution}
\begin{figure}[htb]
\centering
\includegraphics[width=0.5\textwidth]{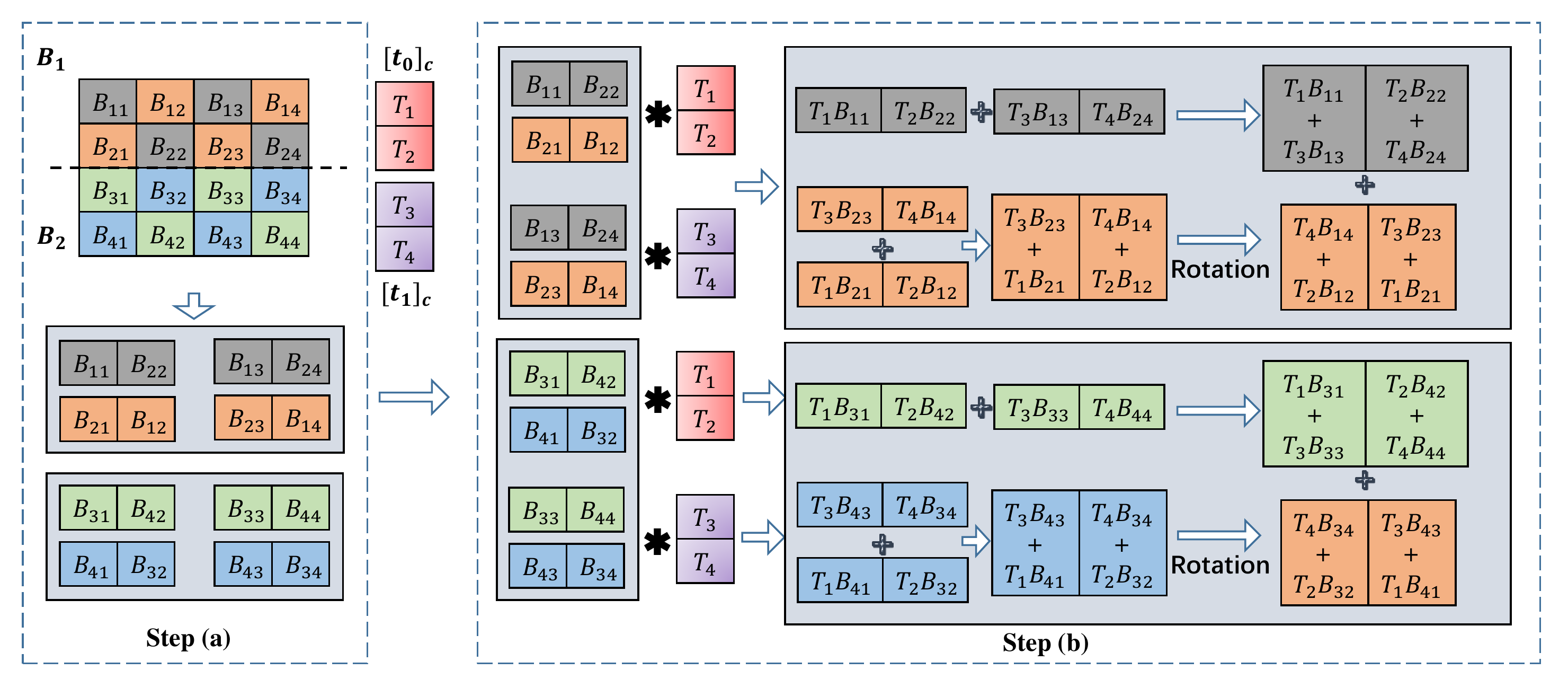}
\caption{Our MIMO convolution}
\label{Fig: Our MIMO convolution}
\end{figure}

We present an improved solution to execute the MIMO convolution based on the  GAZELLE's method. Our key trick is to design strategies to significantly reduce the number of rotation operations involved in computing the convolution between each of the  $\frac{c_0c_i}{c_n^2}$ blocks and the corresponding input channels. To be precise, we find that for each of the $\frac{c_0c_i}{c_n^2}$ blocks, GAZELLE's method requires $(c_n-1)$ rotation operations to obtain the convolution for that block. However, since our goal is to obtain the convolution of each kernel, we  actually do not need to get the convolution of each block. We propose a \textit{first-Add-second-Rotation method} to reduce the complexity of  rotations. Specifically, as shown in \textbf{Figure}~\ref{Fig: Our MIMO convolution}, the entire kernel block is divided into two blocks, \textit{i.e.}, $\mathbf{B_1}$ and $\mathbf{B_2}$, where each block is the combination of $\frac{c_i}{c_n}c_n$-by-$c_n$ divided blocks, which corresponding to the same kernels (The first two kernels are contained in block $\mathbf{B_1}$, while the third and fourth kernel are organized in $\mathbf{B_1}$).

For each newly constructed block, all its row vectors are first convolved with the corresponding input ciphertexts in a SISO manner. Then, those convolved vectors associated with the same kernel order will be added together (see add operation of Step (b) in \textbf{Figure}~\ref{Fig: Our MIMO convolution} before executing rotations). Finally, these added vectors are rotated to the same kernel order and summed to obtain the final convolution result. Our method incurs $(c_n-1)$ rotations for each of $\frac{c_o}{c_n}$ newly constructed blocks, which reduces the computation complexity with respect to rotation by a factor of $\frac{c_i}{c_n}$ compared to the method proposed by GAZELLE. This is substantially because in state-of-the-art DNN networks such as ResNets, $\frac{c_i}{c_n}$ can reach 256. This is mainly due to the fact that neural networks contain a large number of large-sized feature maps in order to capture complex input features.

We now analyze the computation complexity of our proposed method. Specifically,  for each of the $\frac{c_o}{c_n}$ blocks, there are a total of $c_i$ SISO convolutions. Additionally, there are $\frac{c_i}{c_n}$  convolutions summed for each of the $c_n$ kernel orders, which introduces $c_N$ added convolutions. These added convolutions will be further rotated to the same kernels and then summed up to obtain the final ciphertext result. Therefore, our method requires  a total of $\frac{c_0}{c_n}(c_n-1)+\frac{c_i(k_wk_h-1)}{c_n}$ rotations, $k_wk_h\frac{c_ic_o}{c_n}$ ScMult and $\frac{c_o}{c_n}(c_ik_wk_h-1)$ Add operations. This method also outputs $\frac{c_o}{c_n}$ ciphertexts.

Table~\ref{Complexity of Each of Method} shows the comparison of our method with existing methods in computation complexity. It is obvious that our method has better complexity, especially for rotation operations. Our experiments  demonstrate that this boost substantially speeds up the homomorphic operation of the linear layer.

\begin{table}[h]
\centering
\scriptsize
\caption{Computation complexity of each method for convolution}
\label{Complexity of Each of Method}
\begin{tabular}{c|c|c|c}
\Xhline{1pt}
{Method}&\#Rotation&\#ScMul &\#Add\\
\Xhline{1pt}
GAZELLE&$\frac{c_oc_i}{c_n^2}(c_n-1)+\frac{c_i(k_wk_h-1)}{c_n}$& $k_wk_h\frac{c_ic_o}{c_n}$& $\frac{c_o}{c_n}(c_ik_wk_h-1)$\\
\hline
Our method&$\frac{c_0}{c_n}(c_n-1)+\frac{c_i(k_wk_h-1)}{c_n}$& $k_wk_h\frac{c_ic_o}{c_n}$& $\frac{c_o}{c_n}(c_ik_wk_h-1)$\\
\Xhline{1pt}
\end{tabular}
\end{table}

\section{Proof of Theorem 4.1}
\label{proof of th1}
\begin{proof}
Similar to SIMC, let $\pi^{f}_{\mathtt{Non-lin}}$ denote the designed protocol for computing the non-linear layer, and $\mathcal{F}^{f}_{\mathtt{Non-lin}}$ denote  the  functionality of our protocol (see \textbf{Figure}~\ref{Functionality of the nonlinear layer}). Also, we use $\mathtt{Real}$ to represent the real protocol execution of $\pi^{f}_{\mathtt{Non-lin}}$ between $S_0$ and the adversary $\mathcal{A}$ compromising $S_1$. We will prove that the view of $\mathcal{A}$  in $\mathtt{Real}$ is indistinguishable from the  view in the simulated execution $\mathtt{Sim}$.  We exploit the standard hybrid argument to prove this, where we define three intermediate hybrids, \textit{i.e.}, $\mathtt{Hyb_1}$, $\mathtt{Hyb_2}$ and $\mathtt{Hyb_3}$, and show the indistinguishability of this hybrids in consecutive executions.

$\mathtt{Hyb_1}$: This hybrid execution is identical to $\mathtt{Real}$ except in the authentication phase.  To be precise, in the authentication phase, the simulator $\mathcal{S}$ use labels $\hat{\mathtt{lab}_{i, j}^{out}}$ (described below) to replace the labels $\mathtt{lab}_{i, j}^{out}$ used in $\mathtt{Real}$. Please note that in this hybrid the simulator $\mathcal{S}$ can access $S_0$' input $\left \langle \mathbf{u} \right \rangle_0$ and $\alpha$, where $\left \langle \mathbf{u} \right \rangle_0 +\left \langle \mathbf{u} \right \rangle_1= \mathbf{u}$. Let $\delta=(\mathbf{u}||sign(\mathbf{u}))$. Therefore, for $i\in[2\kappa]$, we set  $\hat{\mathtt{lab}_{i, j}^{out}}=\mathtt{lab}_{i, j}^{out}$ if $j=\delta[i]$, otherwise, $\hat{\mathtt{lab}_{i, 1-\delta[i]}^{out}}$ (\textit{i.e.}, the ``other" label) is set to a random value chosen from $\{0, 1\}^{\lambda}$ uniformly, where the first bit of  $\hat{\mathtt{lab}_{i, 1-\delta[i]}^{out}}$ is $1-\xi_{i, \delta[i]}$. We provide the formal description of $\mathtt{Hyb_1}$ as follows, where the  indistinguishability between the view of $\mathcal{A}$ in $\mathtt{Real}$ and $\mathtt{Hyb_1}$ is directly derived from the authenticity of the garbled circuit (see Section~\ref{Garbled Circuits}).
\begin{itemize}
    \item[1.] $\mathcal{S}$ receives $\left \langle \mathbf{u} \right \rangle_1$ from $\mathcal{A}$ as the input of OT$_{\lambda}^{\kappa}$.
\item[2.] Garbled Circuit Phase:
\begin{itemize}
\item For the boolean circuit $booln^f$, $\mathcal{S}$ first computes $\mathtt{Garble}(1^\lambda, booln^f)\rightarrow (\mathtt{GC}, \{ \{\mathtt{lab}_{i,j}^{in}\},\{\mathtt{lab}_{i,j}^{out}\}\}_{j\in\{0,1\}})$ for each $i\in[2\kappa]$,  and then for ${i\in\{\kappa+1, \cdots, 2\kappa\}}$ sends  $\{\mathtt{\tilde{lab}}_{i}^{in}=\mathtt{lab}_{i, \left \langle \mathbf{u} \right \rangle_1[i]}^{in}\}$ to $\mathcal{A}$ as the output of  OT$_{\lambda}^{\kappa}$. In addition, $\mathcal{S}$ sends the garbled  circuit $\mathtt{GC}$ and its garbled inputs $\{ \{\mathtt{\tilde{lab}}_{i}^{in}=\mathtt{lab}_{i, \left \langle \mathbf{u} \right \rangle_0[i]}\}_{i\in[\kappa]}$ to $\mathcal{A}$.
\end{itemize}
\item[3.] Authentication Phase 1:
\begin{itemize}
\item $\mathcal{S}$ sets $\delta=(\mathbf{u}||sign(\mathbf{u}))$.
\item For $i\in[2\kappa]$, $\mathcal{S}$ sets $\hat{\mathtt{lab}_{i, j}^{out}}=\mathtt{lab}_{i, \delta[i]}^{out}$ if $j=\delta[i]$.
\item For $i\in[2\kappa]$, if $j=1-\delta[i]$, $\mathcal{S}$ sets $\hat{\mathtt{lab}_{i, j}^{out}}$ as  a random value chosen from $\{0, 1\}^{\lambda}$ uniformly, where first bit of  $\hat{\mathtt{lab}_{i, 1-\delta[i]}^{out}}$ is $1-\xi_{i, \delta[i]}$.
\item $\mathcal{S}$ computes and sends $\{ct_{i,j}, \hat{ct_{i, j}}\}_{i\in [\kappa], j\in\{0, 1\}}$ to $\mathcal{A}$ using $\hat{\mathtt{lab}_{i, j}^{out}}_{i\in [2\kappa], j\in\{0, 1\}}$. This process is same as  in $\mathtt{Real}$ execution using  ${\mathtt{lab}_{i, j}^{out}}_{i\in [2\kappa], j\in\{0, 1\}}$.
\end{itemize}
\item[4.]Local Computation Phase: The execution of this phase is indistinguishable from $\mathtt{Real}$ since no information needs to be exchanged between $\mathcal{S}$ and $\mathcal{A}$.
\item[5.]  Authentication Phase 2:
\begin{itemize}
\item $\mathcal{S}$ randomly  selects a fresh shares of authenticated Beaver's multiplicative triple $(A, B, C)$, where $C=A\cdot B$.
\item $\mathcal{S}$ sends $\Gamma=\left \langle \mathbf{u} \right \rangle_0-\langle A\rangle_0$ to $\mathcal{A}$.
\end{itemize}
    \end{itemize}

$\mathtt{Hyb_2}$: We will make four changes to $\mathtt{Hyb_1}$ to obtain $\mathtt{Hyb_2}$, and argue that $\mathtt{Hyb_2}$ is indistinguishable from $\mathtt{Hyb_1}$ from the adversary's view. To be precise, let $\mathtt{GCEval}(\mathtt{GC}, \{\mathtt{\tilde{lab}}_{i}^{in}\}_{i\in[2\kappa]})\rightarrow \{(\tilde{\xi}_{i}||\tilde{\zeta}_{i})_{i\in[2\kappa]}= \{\mathtt{\tilde{lab}}_{i}^{out}\}_{i\in[2\kappa]}\}$. First, we have $\{\mathtt{\tilde{lab}}_{i}^{out}= \mathtt{{lab}}_{i, \delta[i]}^{out}\}_{i\in[2\kappa]}$ based on the correctness of garbled circuits.  Second, we note that   ciphertexts $\{ct_{i,1-\tilde{\xi}_{i}}, \hat{ct}_{i, 1-\tilde{\xi}_{i+\kappa}}\}_{i\in [\kappa]}$  are computed by exploiting the ``other" set of output labels picked uniformly in $\mathtt{Hyb_1}$. Based on this observation, $\mathcal{S}$  actually can directly sample them uniformly at random.  Third, in real execution,  for every $i\in [\kappa]$ and $j\in \{0, 1\}$, $S_0$ sends $ct_{i, \xi_{i, j}}$ and $\hat{ct}_{i, \xi_{i+\kappa, j}}$ to $S_1$, and then $S_1$ computes $c_i$, $d_i$ and $e_i$ based on them. To simulate this,  $\mathcal{S}$ only needs to uniformly select random values $c_i$, $d_i$ and $e_i$ which satisfy $\langle \alpha \mathbf{u}\rangle_1=(-\sum_{i\in[\kappa]}c_i2^{i-1})$, $\langle sign(\mathbf{u})\rangle_1=(-\sum_{i\in[\kappa]}d_i2^{i-1})$ and $\langle \alpha sign(\mathbf{u})\rangle_1=(-\sum_{i\in[\kappa]}e_i2^{i-1})$. Finally, since $\langle \alpha \mathbf{u}\rangle_1$, $\langle sign(\mathbf{u})\rangle_1$ and $\langle \alpha sign(\mathbf{u})\rangle_1$ are part the outputs of functionality $\mathcal{F}^{f}_{\mathtt{Non-lin}}$, $\mathcal{S}$  can obtain these as the outputs from $\mathcal{F}^{f}_{\mathtt{Non-lin}}$.  In summary,  with the above changes, $\mathcal{S}$ no longer needs an input $\alpha$ of $S_0$. We provide the formal description of $\mathtt{Hyb_2}$ as follows.
 \begin{itemize}
    \item[1.] $\mathcal{S}$ receives $\left \langle \mathbf{u} \right \rangle_1$ from $\mathcal{A}$ as the input of OT$_{\lambda}^{\kappa}$.
\item[2.] Garbled Circuit Phase: Same as $\mathtt{Hyb_1}$.
\item[3.] Authentication Phase 1:
\begin{itemize}
\item $\mathcal{S}$ runs $\mathtt{GCEval}(\mathtt{GC}, \{\mathtt{\tilde{lab}}_{i}^{in}\}_{i\in[2\kappa]})\rightarrow \{(\tilde{\xi}_{i}||\tilde{\zeta}_{i})_{i\in[2\kappa]}= \{\mathtt{\tilde{lab}}_{i}^{out}\}_{i\in[2\kappa]}\}$.
\item $\mathcal{S}$ learns $\langle \alpha \mathbf{u}\rangle_1$, $\langle sign(\mathbf{u})\rangle_1$ and $\langle \alpha sign(\mathbf{u})\rangle_1$ by sending  $\langle\mathbf{u}\rangle_1$ to $\mathcal{F}^{f}_{\mathtt{Non-lin}}$.
\item For $i\in[\kappa]$, $\mathcal{S}$ uniformly selects random values $c_i$, $d_i$ and $e_i\in \mathbb{F}_{p}$ which satisfy $\langle \alpha \mathbf{u}\rangle_1=(-\sum_{i\in[\kappa]}c_i2^{i-1})$, $\langle sign(\mathbf{u})\rangle_1=(-\sum_{i\in[\kappa]}d_i2^{i-1})$ and $\langle \alpha sign(\mathbf{u})\rangle_1=(-\sum_{i\in[\kappa]}$ $e_i2^{i-1})$.
\item  For every $i\in [\kappa]$, $\mathcal{S}$ computes $ct_{i, \tilde{\xi}_{i}}=c_i \oplus \mathbf{Trun}_\kappa(\tilde{\zeta}_{i})$  and $\hat{ct}_{i, \tilde{\xi}_{i+\kappa}}=(d_i||e_i)\oplus\mathbf{Trun}_{2\kappa}(\tilde{\zeta}_{i+\kappa})$. For ciphertexts $\{ct_{i,1-\tilde{\xi}_{i}}, \hat{ct}_{i, 1-\tilde{\xi}_{i+\kappa}}\}_{i\in [\kappa]}$, $\mathcal{S}$   samples them uniformly at random.
\item $\mathcal{S}$ sends $\{ct_{i,j}, \hat{ct_{i, j}}\}_{i\in [\kappa], j\in\{0, 1\}}$ to $\mathcal{A}$.
\end{itemize}
\item[4.]Local Computation Phase: The execution of this phase is indistinguishable from $\mathtt{Real}$ since no information needs to be exchanged between $\mathcal{S}$ and $\mathcal{A}$.
\item[5.]  Authentication Phase 2:
\begin{itemize}
\item $\mathcal{S}$ randomly  selects a fresh shares of authenticated Beaver's multiplicative triple $(A, B, C)$, where $C=A\cdot B$.
\item $\mathcal{S}$ sends $\Gamma=\left \langle \mathbf{u} \right \rangle_0-\langle A\rangle_0$ to $\mathcal{A}$.
\end{itemize}
    \end{itemize}

$\mathtt{Hyb_3}$: This hybrid we  remove $\mathcal{S}$'s dependence on $S_0$'s input $\langle \mathbf{u} \rangle_0$. The indistinguishability between $\mathtt{Hyb_3}$ and $\mathtt{Hyb_2}$ stems from the security of the garbled circuit. We provide the formal description of $\mathtt{Hyb_3}$ below.
 \begin{itemize}
    \item[1.] $\mathcal{S}$ receives $\left \langle \mathbf{u} \right \rangle_1$ from $\mathcal{A}$ as the input of OT$_{\lambda}^{\kappa}$.
\item[2.] Garbled Circuit Phase:
\begin{itemize}
\item $\mathcal{S}$ samples $\mathtt{Garble}(1^\lambda, booln^f)\rightarrow (\tilde{\mathtt{GC}}, \{\hat{\mathtt{lab}}_{i}^{in}\}_{i\in\{\kappa+1, \cdots, 2\kappa\}})$  and sends $\{\hat{\mathtt{lab}}_i\}_{i\in\{\kappa+1, \cdots, 2\kappa\}}$ to $\mathcal{A}$ as the output of  OT$_{\lambda}^{\kappa}$. $\mathcal{S}$ also sends $\tilde{\mathtt{GC}}$ and $\{\hat{\mathtt{lab}}_{i}^{in}\}_{i\in[\kappa]}$ to $\mathcal{A}$.
\end{itemize}
\item[3.] Authentication Phase 1:
\item[4.]Local Computation Phase: The execution of this phase is indistinguishable from $\mathtt{Real}$ since no information needs to be exchanged between $\mathcal{S}$ and $\mathcal{A}$.
\item[5.]  Authentication Phase 2: Same as $\mathtt{Hyb_2}$, where $\mathcal{S}$ uses $(\langle \mathbf{u}\rangle_1, \tilde{\mathtt{GC}},$ and $ \{\hat{\mathtt{lab}}_{i}^{in}\}_{i\in[2\kappa]})$ to process this phase for $\mathcal{A}$.
\begin{itemize}
\item $\mathcal{S}$ randomly  selects a fresh shares of authenticated Beaver's multiplicative triple $(A, B, C)$, where $C=A\cdot B$.
\item $\mathcal{S}$ randomly selects  $\hat{\Gamma}$ uniformly from $\mathbb{F}_p$  and send it to $\mathcal{A}$.
\end{itemize}
    \end{itemize}
\end{proof}
\section{Protocol of Secure Inference}
\label{protocol of secure infernece}
In this section, we describe the details of our secure inference protocol (called $\pi_{inf}$). For simplicity, suppose a neural network (NN) consists of alternating linear and nonlinear layers. Let the specifications of the linear layer be $\mathbf{N}_1$, $\mathbf{N}_2$, $\cdots$, $\mathbf{N}_m$, and the nonlinear layer be $f_1$, $f_2$, $\cdots$, $f_{m-1}$. Given an input vector $\mathbf{t}_0$, one needs to  sequentially compute $\mathbf{u}_i=\mathbf{N}_i\cdot \mathbf{t}_{i-1}, \mathbf{t}_i= f_i(\mathbf{u}_i)$, where $i\in[m-1]$. As a result, we have $\mathbf{u}_m=\mathbf{N}_m\cdot\mathbf{t}_{m-1}=\mathrm{NN}(\mathbf{t}_0)$. In secure inference, the server $S_0$'s input is weights of all linear layer, \textit{i.e.}, $\mathbf{N}_1, \cdots, \mathbf{N}_m$  while the input of $S_1$ is $\mathbf{t}$.  The goal of $\pi_{inf}$ is to learn $\mathrm{NN}(\mathbf{t}_0)$ for the client $S_1$. We provide the overview of $\pi_{inf}$ below and give the details of protocol in \textbf{Figure}~\ref{Protocol of secure inference}.
\renewcommand\tablename{Figure}
\renewcommand \thetable{\arabic{table}}
\setcounter{table}{12}
\setcounter{figure}{12}
\begin{table*}[htb]
\centering
\begin{tabular}{|p{17.4cm}|}
\hline
\textbf{Preamble}: Consider a neural network (NN) consists of $m$  linear layers and $m-1$  nonlinear layers.  Let the specifications of the linear layer be $\mathbf{N}_1$, $\mathbf{N}_2$, $\cdots$, $\mathbf{N}_m$, and the nonlinear layer be $f_1$, $f_2$, $\cdots$, $f_{m-1}$.\\
 \textbf{Input:}$S_0$ holds $\{ \mathbf{N}_j\in\mathbb{F}_p^{n_j\times n_{j-1}}\}_{j\in[m]}$, \textit{i.e.}, weights for the $m$ linear layers. $S_1$  holds $\mathbf{t}_0\in \mathbb{F}_p^{n_0}$ as the input of NN.\\
 \textbf{Output:} $S_1$ obtains $\mathrm{NN}(\mathbf{t}_0)$.\\
\textbf{Protocol}:\\
\begin{itemize}
\item[1.] $S_0$ uniformly  selects a random MAC key $\alpha$ from $\mathbb{F}_p$ to be used throughout the protocol execution.
\item[2.] \textbf{First Linear Layer}: $S_0$ and $S_1$ execute the function \textbf{InitLin} to learn $( \langle \mathbf{u}_1 \rangle_b, \langle \mathbf{r}_1 \rangle_b)$ for $b\in\{0, 1\}$, where $S_0$'s inputs for \textbf{InitLin} are $(\mathbf{N}_1, \alpha)$ while $S_1$'s input is $\mathbf{t}_0$.
\item[3.] For each $j\in[m-1]$,
\begin{itemize}
\item [] \textbf{Non-Linear Layer $f_j$}:  $S_0$ and $S_1$ execute the function $\pi^{f_j}_{\mathtt{Non-lin}}$ to learn $( \langle \mathbf{k}_j \rangle_b,  \langle \mathbf{t}_j \rangle_b, \langle \mathbf{d}_j \rangle_b)$ for $b\in\{0, 1\}$, where $S_0$'s inputs for $\pi^{f_j}_{\mathtt{Non-lin}}$ are $(\langle \mathbf{u}_j \rangle_0, \alpha)$ while $S_1$'s input is $\langle \mathbf{u}_j \rangle_1$.
\item[] \textbf{Linear layer $j+1$}: $S_0$ and $S_1$ execute the function \textbf{Lin} to learn $( \langle \mathbf{u}_{j+1} \rangle_b,  \langle \mathbf{r}_{j+1} \rangle_b, \langle \mathbf{z}_{j+1} \rangle_b)$ for $b\in\{0, 1\}$, where $S_0$'s inputs for \textbf{Lin} are $(\langle \mathbf{t}_j \rangle_0, \langle \mathbf{d}_j \rangle_0, \mathbf{N}_{j+1}, \alpha)$ while $S_1$'s inputs are $(\langle \mathbf{t}_j \rangle_1, \langle \mathbf{d}_j \rangle_1)$.
\end{itemize}
\item[4.] \textbf{Consistency Check}:
\begin{itemize}
\item For $j\in[m-1]$, $S_0$ selects $\mathbf{s}_j\in_{R}\mathbb{F}_{p}^{n_j}$ and $\mathbf{s}'_j\in_{R}\mathbb{F}_{p}^{n_{j+1}}$. $S_0$ sends $(\mathbf{s}_j, \mathbf{s}'_j)$ to $S_1$.
\item $S_1$ computes $\langle \mathbf{q} \rangle_1=\sum_{j\in[m-1]}\left(  (\langle \mathbf{r}_j \rangle_1-\langle \mathbf{k}_j \rangle_1)\cdot \mathbf{s}_j+\langle \mathbf{z}_{j+1} \rangle_1\cdot\mathbf{s}'_j \right)$ and sends it to $S_0$.
\item $S_0$ computes $\langle \mathbf{q} \rangle_0=\sum_{j\in[m-1]}\left(  (\langle \mathbf{r}_j \rangle_0-\langle \mathbf{k}_j \rangle_0)\cdot \mathbf{s}_j+\langle \mathbf{z}_{j+1} \rangle_0\cdot\mathbf{s}'_j \right)$.
\item $S_0$ aborts if  $\langle \mathbf{q} \rangle_0+\langle \mathbf{q} \rangle_1 \mod p\neq 0$. Else, sends $\langle \mathbf{u}_{m} \rangle_0$ to $S_1$.
\end{itemize}
\item[5.] \textbf{Output Phase}: $S_1$ outputs $\langle \mathbf{u}_{m} \rangle_0+\langle \mathbf{u}_{m} \rangle_1\mod p$ if $S_0$ didn't abort in the previous step.
\end{itemize}\\
\hline
\end{tabular}
\caption{Secure inference protocol $\pi_{inf}$}
\label{Protocol of secure inference}
\end{table*}

Our protocol can be roughly divided into two phases: the evaluation phase and the consistency check phase. We perform the computation of alternating linear and nonlinear layers with appropriate parameters in the evaluation phase. After that, the server performs a consistency check phase to verify the consistency of the calculations so far. The output will be released to client if the check passes.

\begin{itemize}
\item \textbf{ Linear Layer Evaluation}: To evaluate the first linear layer, $S_0$ and $S_1$ execute the function \textbf{InitLin} to learn $( \langle \mathbf{u}_1 \rangle_b, \langle \mathbf{r}_1 \rangle_b)$ for $b\in\{0, 1\}$, where $S_0$'s inputs for \textbf{InitLin} are $(\mathbf{N}_1, \alpha)$ while $S_1$'s input is $\mathbf{t}_0$. This process is to compute the authenticated shares of $\mathbf{u}_1$, \textit{i.e.}, shares of $\mathbf{u}_1$ and $\mathbf{r}_1$. We use $\mathbf{r}_1$ to represent the authentication of $\mathbf{u}_1$. To evaluate the $i$-th linear layer ($i>1$), $S_0$ and $S_1$ execute the function \textbf{Lin} to learn $( \langle \mathbf{u}_{i} \rangle_b,  \langle \mathbf{r}_{i} \rangle_b, \langle \mathbf{z}_{i} \rangle_b)$ for $b\in\{0, 1\}$, where $S_0$'s inputs for \textbf{Lin} are $(\langle \mathbf{t}_{i-1} \rangle_0, \langle \mathbf{d}_{i-1} \rangle_0, \mathbf{N}_{i}, \alpha)$ while $S_1$'s inputs are $(\langle \mathbf{t}_{i-1} \rangle_1, \langle \mathbf{d}_{i-1} \rangle_1)$. We use $\mathbf{d}_{i-1}$ to denote the authentication on $\mathbf{t}_{i-1}$. Therefore, \textbf{Lin} outputs shares of $ \mathbf{u}_{i}= \mathbf{N}_{i} \mathbf{t}_{i-1}$, $ \mathbf{r}_{i}=\mathbf{N}_{i} \mathbf{d}_{i-1}$ and an additional tag $\mathbf{z}_i=(\alpha^3\mathbf{t}_{i}-\alpha^2\mathbf{d}_{i-1})$.
\item \textbf{Non-Linear Layer Evaluation}: To evaluate the $i$-th ($i\in[m-1]$) non-linear layer, $S_0$ and $S_1$ execute the function $\pi^{f_i}_{\mathtt{Non-lin}}$ to learn $( \langle \mathbf{k}_i \rangle_b,  \langle \mathbf{t}_i \rangle_b, \langle \mathbf{d}_i \rangle_b)$ for $b\in\{0, 1\}$, where $S_0$'s inputs for $\pi^{f_i}_{\mathtt{Non-lin}}$ are $(\langle \mathbf{u}_i \rangle_0, \alpha)$ while $S_1$'s input is $\langle \mathbf{u}_i \rangle_1$, where we use $\mathbf{k}_i$ to represent another set of shares of authentication on $\mathbf{u}_i$.
\item \textbf{Consistency Check Phase}: The server performs the following process to check the correctness of the calculation.
\begin{itemize}
\item For each $i\in\{2, \cdots, m\}$, verify that the authentication share entered into the function \textbf{Lin} is valid by verifying that $\mathbf{z}_i=0^{n_i-1}$.
\item For each $i\in[m-1]$, check that the shares input of the function $\pi^{f}_{\mathtt{Non-lin}}$ is  same as the output of \textbf{Lin} under the $i$-th linear layer by verifying that $\mathbf{r}_i-\mathbf{k}_i=0^{n_i}$.
\end{itemize}
\end{itemize}
Finally, all of the above checks can be combined into a single check by $S_0$ to pick up random scalars. If the check passes, the final prediction can be reconstructed by $S_1$, otherwise, $S_0$ aborts and returns the final share to $S_1$.\\
\quad \\
\textbf{Correctness}. We briefly describe the correctness of our secure inference protocol $\pi_{inf}$. In more detail, we first have $\mathbf{u}_1=\mathbf{N}_1\cdot \mathbf{t}_{0}$ and $\mathbf{r}_1=\alpha\mathbf{u}_1$ by the correctness of \textbf{InitLin}. Then, for each $i\in\{2, \cdots, m\}$, we have $\mathbf{u}_1=\mathbf{N}_i\cdot \mathbf{t}_{i-1}$, $\mathbf{r}_i=\mathbf{N}_i\cdot \mathbf{d}_{i-1}$, and $\mathbf{z}_i=\alpha^3\mathbf{t}_{i-1}-\alpha^2\mathbf{d}_{i-1}$ by the correctness of \textbf{Lin}. Furthermore, for each $i\in[m-1]$, it holds that $\mathbf{k}_i=\alpha\mathbf{u}_i$, $\mathbf{t}_i=f_i(\mathbf{u}_i)$, and $\mathbf{d}_i=\alpha f_i(\mathbf{u}_i)$ by the correctness of $\pi^{f}_{\mathtt{Non-lin}}$. On the other hand, we can observe that $q=0$ since for each $i\in[m-1]$, $\mathbf{z}_{i+1}=0$ and $\mathbf{r}_i=\mathbf{k}_i$. Finally, we have $\mathbf{u}_m=NN(\mathbf{t}_0)$.\\
\quad \\
\textbf{Security}. Our secure inference protocol  has the same security properties as SIMC, \textit{i.e.}, it  is secure against the malicious client model. We provide the following theorem.
\setcounter{theorem}{1}
\begin{theorem}
 Our secure inference protocol $\pi_{inf}$ is secure against a semi-honest server $S_0$ and any malicious adversary $\mathcal{A}$ corrupting the client $S_1$.
\end{theorem}
\begin{proof}
We leverage simulation-based security to prove that our protocol is secure against the semi-honest server $S_0$ and the malicious client $S_1$. It is easy to demonstrate the security of our protocol for $S_0$. Specifically, during the evaluation process of the protocol, $S_0$ can only learn the shares disclosed during the evaluation. The simulator can choose random values evenly from the filed to simulate these shares. Similarly, in the consistency check phase, $S_0$ learns $\langle q \rangle_1$, which can also be easily simulated by the simulator using $q=0$, and $\langle q \rangle_0$ can be calculated locally.

\begin{table*}[htb]
\centering
\begin{tabular}{|p{17.4cm}|}
\hline
\textbf{Preamble}: $\mathcal{S}$ interacts with an adversary $\mathcal{A}$ that controls $S_1$, where the input of $S_1$ is $\mathbf{t}_0$. $\mathcal{S}$ sets a flag bit $\mathtt{flag}=0$.\\
\begin{itemize}
\item \textbf{First Linear Layer}: $\mathcal{S}$ execute the function \textbf{InitLin} with input $\mathbf{t}'_0$, and  then sends $\langle \mathbf{u}_1 \rangle_1$ and $\langle \mathbf{r}_1 \rangle_1$ to $\mathcal{A}$.
\item For each $j\in[m-1]$,
\begin{itemize}
\item [] \textbf{Non-Linear Layer $f_j$}:  $\mathcal{S}$ execute the function $\pi^{f_j}_{\mathtt{Non-lin}}$ with input $\langle \mathbf{u}'_1 \rangle_1+\Delta_j^{1}$, and then sends uniform $\langle \mathbf{k}_j \rangle_1$, $\langle \mathbf{t}_j \rangle_1$ and $\langle \mathbf{d}_j \rangle_1$ to $\mathcal{A}$.  $\mathcal{S}$ sets a flag bit $\mathtt{flag}=1$ if $\Delta_j^{1}\neq0^{n_j}$.
\item[] \textbf{Linear layer $j+1$}: $\mathcal{S}$ execute the function \textbf{Lin} with inputs $\langle \mathbf{t}'_j \rangle_1=\langle \mathbf{t}_j \rangle_1+\Delta_{j}^2$ and $\langle \mathbf{d}'_j \rangle_1=\langle \mathbf{d}_j \rangle_1+\Delta_{j}^3$, and then sends uniform $\langle \mathbf{u}_j \rangle_1$, $\langle \mathbf{t}_j \rangle_1$, and $\langle \mathbf{z}_j \rangle_1$ to $\mathcal{A}$. In addition,  $\mathcal{S}$ sets a flag bit $\mathtt{flag}=1$ if $(\Delta_j^{2}, \Delta_j^{3})\neq (0^{n_j}, 0^{n_j})$.
\end{itemize}
\item[4.] \textbf{Consistency Check}:
\begin{itemize}
\item For $j\in[m-1]$, $\mathcal{S}$ selects $\mathbf{s}_j\in_{R}\mathbb{F}_{p}^{n_j}$ and $\mathbf{s}'_j\in_{R}\mathbb{F}_{p}^{n_{j+1}}$. $S_0$ sends $(\mathbf{s}_j, \mathbf{s}'_j)$ to $S_1$.
\item $\mathcal{A}$ computes $\langle \mathbf{q} \rangle_1=\Delta^4+\sum_{j\in[m-1]}\left(  (\langle \mathbf{r}_j \rangle_1-\langle \mathbf{k}_j \rangle_1)\cdot \mathbf{s}_j+\langle \mathbf{z}_{j+1} \rangle_1\cdot\mathbf{s}'_j \right)$ and sends it to $\mathcal{S}$. $\mathcal{S}$ set $\mathtt{flag}=1$ if $\Delta^4\neq0$.
\end{itemize}
\item[5.] \textbf{Output Phase}: $\mathcal{S}$ obtains NN$(\mathbf{t}'_0)$ and send NN$(\mathbf{t}'_0)-\langle \mathbf{u}_m \rangle_1$ to $\mathcal{A}$ if $\mathtt{flag}=0$. Otherwise, it sends abort to both $\pi_{inf}$ and $\mathcal{A}$.
\end{itemize}\\
\hline
\end{tabular}
\caption{Simulator against malicious client for $\pi_{inf}$}
\label{Simulation of secure inference}
\end{table*}

 we now demonstrate the security of the proposed protocol against malicious clients. Review the ability of a malicious client, \textit{i.e.}, it can arbitrarily violate the configuration of the protocol,  especially,  send inconsistent input during the execution of the protocol. To be precise, adversary $\mathcal{A}$ can do the following actions: (1) Invoke \textbf{InitLin} with input $\mathbf{t}'_0\neq \mathbf{t}_0$ (client's original input). and learns $\langle \mathbf{u}_1 \rangle_1$ and $\langle \mathbf{r}_1 \rangle_1$; (2) For each $j\in[m-1]$,  $\mathcal{A}$ can invoke $\pi^{f_j}_{\mathtt{Non-lin}}$ with input $\langle \mathbf{u}'_1 \rangle_1+\Delta_j^{1}$, and  learn $\langle \mathbf{k}_j \rangle_1$, $\langle \mathbf{t}_j \rangle_1$ and $\langle \mathbf{d}_j \rangle_1$; (3) For each $j\in[m-1]$, $\mathcal{A}$ can invoke \textbf{Lin} with inputs $\langle \mathbf{t}'_j \rangle_1=\langle \mathbf{t}_j \rangle_1+\Delta_{j}^2$ and $\langle \mathbf{d}'_j \rangle_1=\langle \mathbf{d}_j \rangle_1+\Delta_{j}^3$, and learn $\langle \mathbf{u}_j \rangle_1$, $\langle \mathbf{t}_j \rangle_1$, and $\langle \mathbf{z}_j \rangle_1$; (4) $\mathcal{A}$ can send  an incorrect share of $q$ to $S_0$ by adding an error $\Delta^4$ to his share. Using the above notation, $\mathcal{A}$ is considered to be honestly following the configuration of the protocol if all $\Delta$'s are 0. We provide a formal simulation of the view of adversary in \textbf{Figure}~\ref{Simulation of secure inference}.

  Based on this, we analyze the following two cases: In the first case, $\mathcal{A}$ did not violate the protocol. It is easy to deduce that the simulated view is exactly the same as the real view. When $\mathcal{A}$ violates the execution of the protocol, we prove that there is a non-zero $\Delta$, so that $\mathcal{S}$ returns abort to the adversary $\mathcal{A}$ with a probability of 1. This is mainly due to the attributes of authentication shares, which guarantees the immutability of data under authentication sharing. Readers  can refer to  SIMC \cite{chandran2021simc} for more details, which gives a detailed analysis probabilistically.
\end{proof}

\end{document}